\documentclass{article} % Sets the basic 'template' for the general layout of the document.

\usepackage[utf8]{inputenc} % Specifies TeX to process the symbols you type as unicode symbols ("utf8" stands for "Unicode Transformation Format 8")
\usepackage{amsmath} % Provides many features for documents containing mathematical formulas
\usepackage{amsthm} % Helps to define theorem-like structures
\usepackage{todonotes}
\usepackage{stix2} % Provides extended math symbol and font collections
\usepackage{caption}
\usepackage{subcaption}

\usepackage[margin=1in]{geometry}
\usepackage{xcolor}
\usepackage{enumerate}
\usepackage{graphicx}
\usepackage{tikz}
\usepackage[colorlinks]{hyperref}
\usepackage[nameinlink,capitalise]{cleveref}
\usepackage[style=alphabetic]{biblatex}
\usepackage{optidef}
\usepackage[htt]{hyphenat} % https://tex.stackexchange.com/questions/299/how-to-get-long-texttt-sections-to-break

\usepackage{pgfplots}
\pgfplotsset{compat=1.17}
\usepgfplotslibrary{fillbetween}

\usepackage{thmtools}
\usepackage{thm-restate}

\allowdisplaybreaks

\makeatletter
\newcommand\footnoteref[1]{\protected@xdef\@thefnmark{\ref{#1}}\@footnotemark}
\makeatother

\newtheorem{theorem}{Theorem}[section] % This is where you add environment labels. The first brackets contain what you will put in the environment code. The second label is what will be outputted in the PDF.
\newtheorem{corollary}[theorem]{Corollary}

\newtheorem{definition}[theorem]{Definition}
\newtheorem{lemma}[theorem]{Lemma}

\newcommand{\USM}{{\texttt{USM}}}
\newcommand{\CSM}{{\texttt{CSM}}}
\newcommand{\RUSM}{{\texttt{RegularizedUSM}}}
\newcommand{\RCSM}{{\texttt{RegularizedCSM}}}
\newcommand{\DetDG}{{\texttt{DeterministicDG}}}
\newcommand{\RanDG}{{\texttt{RandomizedDG}}}
\newcommand{\RR}{\mathbb R}
\newcommand{\EE}{\mathbb E}
\newcommand{\CE}{\mathcal E}
\newcommand{\NN}{\mathcal N}
\newcommand{\II}{\mathcal I}
\newcommand{\CC}{\mathcal C}
\newcommand{\MM}{\mathcal M}
\newcommand{\PP}{\mathcal P}
\newcommand{\FF}{\mathcal F}
\newcommand{\GG}{\mathcal G}

\newcommand{\TT}{\mathcal T}
\newcommand{\ZZ}{\mathcal Z}

\newcommand{\bone}{\mathbf 1}
\newcommand{\BOPT}{\mathbf{OPT}}
\newcommand{\bx}{\mathbf x}
\newcommand{\by}{\mathbf y}
\newcommand{\bz}{\mathbf z}

\newcommand{\paren}[1]{\left( #1 \right)}
\newcommand{\brac}[1]{\left[ #1 \right]}
\newcommand{\eps}{\epsilon}
\newcommand{\MO}{\mathcal O}
\newcommand{\abs}[1]{\left| #1 \right|}

\newcommand{\BIGO}[1]{\MO\left(#1\right)}

\newcommand{\ceil}[1]{\left\lceil #1 \right\rceil}
\numberwithin{equation}{section} % https://tex.stackexchange.com/questions/54241/change-the-type-of-equation-numbering-in-document-class-article

\title{On Maximizing Sums of Non-monotone Submodular and Linear Functions}
\author{Benjamin Qi\thanks{Department of Electrical Engineering and Computer Science, Massachusetts Institute of Technology; \href{mailto:bqi343@mit.edu}{bqi343@mit.edu}}}
\date{\today}
 
\begin{document} % Only what you write between here and \end{document} will be printed in the pdf

\maketitle % Prints the title, author, and date

\thispagestyle{empty}

\begin{abstract}
We study the problem of \texttt{Regularized Unconstrained Submodular Maximization} (\RUSM{}) as defined by Bodek and Feldman [BF22]. In this problem, you are given a non-monotone non-negative submodular function $f:2^{\NN}\to \mathbb R_{\ge 0}$ and a linear function $\ell:2^{\NN}\to \mathbb R$ over the same ground set $\NN$, and the objective is to output a set $T\subseteq \NN$ approximately maximizing the sum $f(T)+\ell(T)$. Specifically, an algorithm is said to provide an \textit{$(\alpha,\beta)$-approximation} for \RUSM{} if it outputs a set $T$ such that $\EE[f(T)+\ell(T)]\ge \max_{S\subseteq \NN}[\alpha \cdot f(S)+\beta\cdot \ell(S)]$. We also study the setting where $S$ and $T$ are subject to a matroid constraint, which we refer to as \texttt{Regularized \textit{Constrained} Submodular Maximization} (\RCSM{}).

For both \RUSM{} and \RCSM{}, we provide improved $(\alpha,\beta)$-approximation algorithms for the cases of non-positive $\ell$, non-negative $\ell$, and unconstrained $\ell$. In particular, for the case of unconstrained $\ell$, we are the first to provide nontrivial $(\alpha,\beta)$-approximations for \RCSM{}, and the $\alpha$ we obtain for \RUSM{} is superior to that of [BF22] for all $\beta\in (0,1)$.

In addition to approximation algorithms, we provide improved inapproximability results for all of the aforementioned cases. In particular, we show that the $\alpha$ our algorithm obtains for \RCSM{} with unconstrained $\ell$ is tight for $\beta\ge \frac{e}{e+1}$. We also show 0.478-inapproximability for maximizing a submodular function where $S$ and $T$ are subject to a cardinality constraint, improving the long-standing 0.491-inapproximability result due to Gharan and Vondrak [GV10].

\medskip

\end{abstract}

\newpage
\setcounter{page}{1}

\section{Introduction}\label{sec:intro}

\paragraph*{Submodularity.}

\textit{Submodularity} is a property satisfied by many fundamental set functions, including coverage functions, matroid rank functions, and
directed cut functions. Optimization of submodular set functions has found a wealth of applications in machine learning, including the spread of influence in social networks \cite{kempe2003maximizing}, sensor placement \cite{krause2008near}, information gathering \cite{krause2011submodularity}, document summarization \cite{lin2011class,wei2013using,gygli2015video}, image segmentation \cite{jegelka2011submodularity}, and multi-object tracking \cite{shen2018multiobject}, among others (see \cite{krause2014submodular} for a survey).

\paragraph*{Submodular Maximization.}

Many problems involving maximization of non-negative submodular functions can be classified as either \textit{unconstrained} or \textit{constrained}, which we refer to as \USM{} and \CSM{}, respectively. For \USM{}, the objective is to return any set in the domain of the function approximately maximizing the function, while for \CSM{}, the returned set must additionally satisfy a \textit{matroid independence} constraint (or ``matroid constraint'' for short). The simplest nontrivial example of a matroid constraint is a \textit{cardinality} constraint, which means that an upper bound is given on the allowed size of the returned set.

In general, it is impossible to approximate the maxima of instances of \USM{} or \CSM{} to arbitrary accuracy in polynomial time, so we focus on both finding algorithms that return a set with expected value at least $\alpha$ times that of the optimum, known as \textit{$\alpha$-approximation algorithms}, and proving that no such polynomial-time algorithms can exist, known as \textit{$\alpha$-inapproximability results}.

Now we briefly review past results for both \USM{} and \CSM{}. A $(1-e^{-1})$-approximation for monotone \CSM{} was achieved by Nemhauser et al. \cite{nemhauser1978analysis} using a \textit{greedy} algorithm for the special case of a cardinality constraint and later generalized by Calinescu et al. \cite{calinescu2011maximizing} to a matroid constraint using a \textit{continuous greedy} algorithm. On the other hand, a $0.5$-approximation for non-monotone \USM{} was provided by Buchbinder et al. \cite{buchbinder2012double} using a \textit{randomized double greedy} algorithm, while the best known approximation factor for non-monotone \CSM{} is $0.385$ due to Buchbinder and Feldman \cite{buchbinder2016nonsymmetric} using a \textit{local search} followed by an \textit{aided measured continuous greedy}. 

The first two approximation factors are tight; ($1-e^{-1}+\eps$)-inapproximability and ($0.5+\eps$)-inapproximability for any $\eps>0$ were shown by Nemhauser and Wolsey \cite{nemhauser1978best} and Feige et al. \cite{feige2011maximizing}, respectively, using ad hoc methods. On the other hand, the best known inapproximability factor for non-monotone \CSM{} is $0.478$ due to Gharan and Vondrak \cite{gharan2011submodular} using the \textit{symmetry gap} technique of \cite{vondrak2011symmetry}. This technique has the advantage of being able to succinctly reprove the inapproximability results of \cite{nemhauser1978best,feige2011maximizing}.

\paragraph*{Submodular + Linear Maximization.}

In this work we consider approximation algorithms for maximizing the sum of a non-negative non-monotone submodular function $f$ and a linear function $\ell$. The function $g=f+\ell$ is still submodular, though not necessarily non-negative. Here, the linear term has several potential interpretations. For example, by setting $\ell$ to be non-positive, $\ell$ serves as a \textit{regularizer} or \textit{soft constraint} that favors smaller sets \cite{harshaw2019submodular}.

Sviridenko et al. \cite{sviridenko2017optimal} were the first to study algorithms for $f+\ell$ sums in the case of $f$ monotone, in order to provide improved approximation algorithms for monotone \texttt{CSM} with bounded \textit{curvature}. Here, the curvature $c\in [0,1]$ of a non-negative monotone submodular function $g$ is roughly a measure of how far $g$ is from linear. They provide a $(1-c/e-\eps)$-approximation algorithm and a complementary $(1-c/e+\eps)$-inapproximability result. The idea of the algorithm is to decompose $g$ into $f+\ell$ and show that an approximation factor of $1-e^{-1}$ can be achieved with respect to $f$ and an approximation factor of $1$ can be achieved with respect to $\ell$ simultaneously. Formally, if $\II$ is the independent set family of a matroid, the algorithm computes a set $T\in \II$ that satisfies $\EE[g(T)]=\EE[f(T)+\ell(T)]\ge \max_{S\in \II}[(1-e^{-1}-\eps)f(S)+ (1-\eps)\ell(S)]$ by first ``guessing'' the value of $\ell(S)$, and then running continuous greedy. The algorithm also works when the sign of $\ell$ is unconstrained. Feldman subsequently removed the need for the guessing step and the dependence on $\eps\ell(S)$ by introducing a \textit{distorted objective} \cite{feldman2018guess}. Many faster algorithms for the case of $f$ monotone have since been developed \cite{harshaw2019submodular,kazemi2021regularized,nikolakaki2021efficient}. 

However, only very recently has the case of $f$ non-monotone been considered. Lu et al. \cite{lu2021regularized} were the first to do so using a \textit{distorted measured continuous greedy}, showing how to compute $T\in \II$ such that $\EE[f(T)+\ell(T)]\ge \max_{S\in \II}[(e^{-1}-\eps)f(S)+\ell(S)]$, but only when $\ell$ is non-positive. Bodek and Feldman \cite{bodek2022maximizing} were the first to consider the case where $f$ non-monotone and $\ell$ is unconstrained. They define and study the problem of \texttt{Regularized Unconstrained Submodular Maximization} (\RUSM{}):

\begin{definition}[\RUSM{}]

Given a (not necessarily monotone) non-negative submodular function $f:2^{\NN}\to \mathbb R_{\ge 0}$ and a linear function $\ell:2^{\NN}\to \mathbb R$ over the same ground set $\NN$, an algorithm is said to provide an $(\alpha,\beta)$-approximation for \RUSM{} if it outputs a set $T\subseteq \NN$ such that $\EE[f(T)+\ell(T)]\ge \max_{S\subseteq \NN}[\alpha \cdot f(S)+\beta\cdot \ell(S)]$.

\end{definition}

They note that the algorithm of \cite{sviridenko2017optimal} may be modified to provide $(1-e^{-\beta}-\eps,\beta-\eps)$-approximations for $f$ monotone for all $\beta\in [0,1]$. They also note that when $\ell$ is non-positive, the algorithm of \cite{lu2021regularized} provides $(\beta e^{-\beta}-\eps,\beta)$-approximations for all $\beta\in [0,1]$ when $f$ is non-monotone as well as $(1-e^{-\beta}-\eps,\beta)$-approximations for all $\beta\ge 0$ when $f$ is monotone. The main approximation result of \cite{bodek2022maximizing} is the first non-trivial guarantee for \RUSM{} with $f$ non-monotone and the sign of $\ell$ unconstrained. Specifically, they use \textit{non-oblivious local search} to provide $\paren{\alpha(\beta)-\eps,\beta-\eps}\triangleq \paren{\beta(1-\beta)/(1+\beta)-\eps,\beta-\eps}$-approximations for all $\beta\in [0,1]$. They also prove inapproximability results for the cases of $\ell$ non-negative and $\ell$ non-positive using the symmetry gap technique \cite{vondrak2011symmetry}, including $(1-e^{-\beta}+\eps,\beta)$-inapproximability for monotone $f$ and non-positive $\ell$ for all $\beta\ge 0$, showing that the algorithm of \cite{lu2021regularized} is tight for this case \cite[Theorem 1.1]{bodek2022maximizing}.

\section{Our Contributions}\label{sec:our-contrib}

In this work, we present improved approximability and inapproximability results for \RUSM{} as well as the setting where $S$ and $T$ are subject to a matroid constraint, which we refer to as \texttt{Regularized \textit{Constrained} Submodular Maximization} (\RCSM{}):

\begin{definition}[\RCSM{}]

Given a (not necessarily monotone) non-negative submodular function $f:2^{\NN}\to \mathbb R_{\ge 0}$ and a linear function $\ell:2^{\NN}\to \mathbb R$ over the same ground set $\NN$, as well as a matroid with family of independent sets denoted by $\II$, an algorithm is said to provide an $(\alpha,\beta)$-approximation for \RCSM{} if it outputs a set $T\in \II$ such that $\EE[f(T)+\ell(T)]\ge \max_{S\in \II}[\alpha \cdot f(S)+\beta\cdot \ell(S)]$.

\end{definition}

In particular, we are the first to present $(\alpha,\beta)$-approximation algorithms for \RCSM{} when $\ell$ is not non-positive, and the $\alpha$ we obtain for \RUSM{} is superior to that of \cite{bodek2022maximizing} for all $\beta\in (0,1)$. To show approximability, the main techniques we use are the \textit{measured continuous greedy} introduced by Feldman et al. \cite{feldman2011unified} and used by \cite{buchbinder2016nonsymmetric,lu2021regularized}, the \textit{distorted objective} introduced by Feldman \cite{feldman2018guess} and used by \cite{lu2021regularized}, as well as the ``guessing step'' of \cite{sviridenko2017optimal}. To show inapproximability, the main technique we use is the symmetry gap of \cite{vondrak2011symmetry}, and most of our symmetry gap constructions are based on those of \cite{gharan2011submodular}.

\paragraph*{Organization of the Paper.} We present the definitions and notation used throughout this paper in \Cref{sec:prelims}. \Cref{sec:cardinality-inapprox,sec:non-positive,sec:non-neg-usm,sec:non-neg-csm,sec:unconstrained} form the bulk of our paper and are summarized below. We conclude with discussion of open problems in \Cref{sec:open}.

\subsubsection*{\Cref{sec:cardinality-inapprox}: Inapproximability of Maximization with Cardinality Constraint}

We first consider \CSM{} without a regularizer. Gharan and Vondrak \cite{gharan2011submodular} proved 0.491-\allowbreak inapproximability of \CSM{} in the special case where the matroid constraint is a cardinality constraint. We improve the inapproximability factor to 0.478 in \Cref{thm:0.478-inapprox} by modifying a construction from the same paper \cite[Theorem E.2]{gharan2011submodular} that uses the \textit{symmetry gap} technique of \cite{vondrak2011symmetry}.

\subsubsection*{\Cref{sec:non-positive}: Non-positive $\ell$} 

The results of this section are summarized in \Cref{fig:ell-non-pos}. In \Cref{subsec:non-pos-approx}, we present improved  $(\alpha(\beta),\beta)$\allowbreak-approximations for \RUSM{} for all $\beta\ge 0$ and \RCSM{} for all $\beta\in [0,1]$. Previously, the best known result for both \RUSM{} and \RCSM{} was $\alpha(\beta)=\beta e^{-\beta}-\eps$ due to Lu et al. \cite{lu2021regularized}. This function achieves its maximum value at $\alpha(1)=e^{-1}-\eps>0.367$. We improve the approximation factor for \RCSM{} to $\alpha(1)>0.385$, matching the best known approximation factor for \texttt{CSM} without a regularizer due to Buchbinder and Feldman \cite{buchbinder2016nonsymmetric}. Additionally, we show that larger values of $\alpha(\beta)$ are achievable for \RUSM{} when $\beta>1$. The idea is to combine the ``guessing step'' of Sviridenko et al. \cite{sviridenko2017optimal} with a generalization of the \textit{aided measured continuous greedy} algorithm of Buchbinder and Feldman \cite{buchbinder2016nonsymmetric}.

\begin{figure}
    \centering
    
    \begin{tikzpicture}
    \begin{axis}[
        title={\texttt{RegularizedUSM}, non-positive $\ell$},
        xlabel={coefficient of $\ell$ ($\beta$)},
        ylabel={coefficient of $f$ ($\alpha$)},
        xmin=0, xmax=1.4,
        ymin=0, ymax=1,
        legend pos=north west,
        % ymajorgrids=true,
        % grid style=dashed,
    ]
    \addplot[domain=0:2,color=blue]
    {
        (x <= 1) * x * exp(-x) + (x > 1) * exp(-1)
    };
    \addlegendentry{$(\beta e^{-\beta}-\eps, \beta)$ Approximability \cite{lu2021regularized}}
    
    \addplot[dashed,domain=0:2,color=blue,name path=lower]coordinates{
        (0.0, 0.0)
        (0.01, 0.009)
        (0.02, 0.0181)
        (0.03, 0.0271)
        (0.04, 0.0362)
        (0.05, 0.0452)
        (0.06, 0.0543)
        (0.07, 0.0633)
        (0.08, 0.0724)
        (0.09, 0.0814)
        (0.1, 0.0905)
        (0.11, 0.0978)
        (0.12, 0.1051)
        (0.13, 0.1125)
        (0.14, 0.1198)
        (0.15, 0.1271)
        (0.16, 0.1344)
        (0.17, 0.1418)
        (0.18, 0.1491)
        (0.19, 0.1564)
        (0.2, 0.1637)
        (0.21, 0.1696)
        (0.22, 0.1754)
        (0.23, 0.1813)
        (0.24, 0.1871)
        (0.25, 0.193)
        (0.26, 0.1988)
        (0.27, 0.2047)
        (0.28, 0.2105)
        (0.29, 0.2164)
        (0.3, 0.2222)
        (0.31, 0.2268)
        (0.32, 0.2314)
        (0.33, 0.236)
        (0.34, 0.2406)
        (0.35, 0.2452)
        (0.36, 0.2498)
        (0.37, 0.2544)
        (0.38, 0.259)
        (0.39, 0.2635)
        (0.4, 0.2681)
        (0.41, 0.2716)
        (0.42, 0.2752)
        (0.43, 0.2787)
        (0.44, 0.2822)
        (0.45, 0.2857)
        (0.46, 0.2892)
        (0.47, 0.2927)
        (0.48, 0.2962)
        (0.49, 0.2998)
        (0.5, 0.3033)
        (0.51, 0.3059)
        (0.52, 0.3085)
        (0.53, 0.3111)
        (0.54, 0.3137)
        (0.55, 0.3163)
        (0.56, 0.3189)
        (0.57, 0.3215)
        (0.58, 0.3241)
        (0.59, 0.3267)
        (0.6, 0.3293)
        (0.61, 0.3311)
        (0.62, 0.333)
        (0.63, 0.3348)
        (0.64, 0.3366)
        (0.65, 0.3384)
        (0.66, 0.3403)
        (0.67, 0.3421)
        (0.68, 0.3439)
        (0.69, 0.3458)
        (0.7, 0.3476)
        (0.71, 0.3491)
        (0.72, 0.3507)
        (0.73, 0.3522)
        (0.74, 0.3537)
        (0.75, 0.3551)
        (0.76, 0.3566)
        (0.77, 0.3581)
        (0.78, 0.3596)
        (0.79, 0.3611)
        (0.8, 0.3626)
        (0.81, 0.364)
        (0.82, 0.3655)
        (0.83, 0.367)
        (0.84, 0.3682)
        (0.85, 0.3694)
        (0.86, 0.3706)
        (0.87, 0.3718)
        (0.88, 0.373)
        (0.89, 0.3743)
        (0.9, 0.3755)
        (0.91, 0.3767)
        (0.92, 0.3779)
        (0.93, 0.3788)
        (0.94, 0.3798)
        (0.95, 0.3808)
        (0.96, 0.3817)
        (0.97, 0.3827)
        (0.98, 0.3837)
        (0.99, 0.3846)
        (1.0, 0.3856)
        (1.01, 0.3863)
        (1.02, 0.387)
        (1.03, 0.3878)
        (1.04, 0.3885)
        (1.05, 0.3892)
        (1.06, 0.39)
        (1.07, 0.3907)
        (1.08, 0.3913)
        (1.09, 0.3918)
        (1.1, 0.3924)
        (1.11, 0.3929)
        (1.12, 0.3934)
        (1.13, 0.3939)
        (1.14, 0.3944)
        (1.15, 0.3948)
        (1.16, 0.3951)
        (1.17, 0.3955)
        (1.18, 0.3958)
        (1.19, 0.3962)
        (1.2, 0.3965)
        (1.21, 0.3967)
        (1.22, 0.3969)
        (1.23, 0.3971)
        (1.24, 0.3973)
        (1.25, 0.3975)
        (1.26, 0.3977)
        (1.27, 0.3978)
        (1.28, 0.3978)
        (1.29, 0.3979)
        (1.3, 0.3979)
        (1.31, 0.398)
        (1.32, 0.398)
        (1.33, 0.398)
        (1.34, 0.398)
        (1.35, 0.398)
        (1.36, 0.398)
        (1.37, 0.398)
        (1.38, 0.398)
        (1.39, 0.398)
        (1.4, 0.398)
        (1.41, 0.398)
        (1.42, 0.398)
        (1.43, 0.398)
        (1.44, 0.398)
        (1.45, 0.398)
        (1.46, 0.398)
        (1.47, 0.398)
        (1.48, 0.398)
        (1.49, 0.398)
        (1.5, 0.398)
    };
    \addlegendentry{Approximability (\Cref{thm:rusmNonposExtended})}
    
    \addplot[domain=0:2,color=red]{
        1-exp(-x)
    };
    \addlegendentry{Inapproximability \cite[Theorem 1.1]{bodek2022maximizing}}
    
    \addplot[
        color=red,
        dashed
    ]
    coordinates {
    (0.0000,0.2500)(0.0100,0.2525)(0.0200,0.2550)(0.0300,0.2575)(0.0400,0.2600)(0.0500,0.2625)(0.0600,0.2650)(0.0700,0.2675)(0.0800,0.2700)(0.0900,0.2725)(0.1000,0.2750)(0.1100,0.2775)(0.1200,0.2800)(0.1300,0.2825)(0.1400,0.2849)(0.1500,0.2874)(0.1600,0.2899)(0.1700,0.2924)(0.1800,0.2949)(0.1900,0.2974)(0.2000,0.2998)(0.2100,0.3023)(0.2200,0.3048)(0.2300,0.3073)(0.2400,0.3097)(0.2500,0.3122)(0.2600,0.3147)(0.2700,0.3171)(0.2800,0.3196)(0.2900,0.3220)(0.3000,0.3245)(0.3100,0.3269)(0.3200,0.3294)(0.3300,0.3318)(0.3400,0.3343)(0.3500,0.3367)(0.3600,0.3391)(0.3700,0.3416)(0.3800,0.3440)(0.3900,0.3464)(0.4000,0.3488)(0.4100,0.3513)(0.4200,0.3537)(0.4300,0.3561)(0.4400,0.3585)(0.4500,0.3609)(0.4600,0.3633)(0.4700,0.3657)(0.4800,0.3681)(0.4900,0.3705)(0.5000,0.3728)(0.5100,0.3752)(0.5200,0.3776)(0.5300,0.3799)(0.5400,0.3823)(0.5500,0.3847)(0.5600,0.3870)(0.5700,0.3894)(0.5800,0.3917)(0.5900,0.3941)(0.6000,0.3964)(0.6100,0.3987)(0.6200,0.4010)(0.6300,0.4034)(0.6400,0.4057)(0.6500,0.4080)(0.6600,0.4103)(0.6700,0.4126)(0.6800,0.4149)(0.6900,0.4172)(0.7000,0.4195)(0.7100,0.4217)(0.7200,0.4240)(0.7300,0.4263)(0.7400,0.4285)(0.7500,0.4308)(0.7600,0.4331)(0.7700,0.4353)(0.7800,0.4375)(0.7900,0.4398)(0.8000,0.4420)(0.8100,0.4442)(0.8200,0.4465)(0.8300,0.4486)(0.8400,0.4507)(0.8500,0.4527)(0.8600,0.4547)(0.8700,0.4566)(0.8800,0.4585)(0.8900,0.4603)(0.9000,0.4621)(0.9100,0.4639)(0.9200,0.4655)(0.9300,0.4672)(0.9400,0.4687)(0.9500,0.4703)(0.9600,0.4718)(0.9700,0.4732)(0.9800,0.4746)(0.9900,0.4760)(1.0000,0.4773)(1.0100,0.4786)(1.0200,0.4798)(1.0300,0.4810)(1.0400,0.4821)(1.0500,0.4832)(1.0600,0.4843)(1.0700,0.4853)(1.0800,0.4863)(1.0900,0.4873)(1.1000,0.4882)(1.1100,0.4890)(1.1200,0.4899)(1.1300,0.4907)(1.1400,0.4914)(1.1500,0.4921)(1.1600,0.4928)(1.1700,0.4935)(1.1800,0.4941)(1.1900,0.4947)(1.2000,0.4952)(1.2100,0.4958)(1.2200,0.4962)(1.2300,0.4967)(1.2400,0.4971)(1.2500,0.4975)(1.2600,0.4979)(1.2700,0.4982)(1.2800,0.4985)(1.2900,0.4988)(1.3000,0.4990)(1.3100,0.4992)(1.3200,0.4994)(1.3300,0.4996)(1.3400,0.4997)(1.3500,0.4998)(1.3600,0.4999)(1.3700,0.5000)(1.3800,0.5000)(1.3900,0.5000)(1.4000,0.5000)(1.4100,0.5000)(1.4200,0.5000)(1.4300,0.5000)(1.4400,0.5000)(1.4500,0.5000)(1.4600,0.5000)(1.4700,0.5000)(1.4800,0.5000)(1.4900,0.5000)(1.5000,0.5000)
    };
    \addlegendentry{Inapproximability \cite[Theorem 1.3]{bodek2022maximizing}}
        
    \addplot[name path=upper, color=red, densely dotted]
    coordinates{
        (0.0000,0.0000)(0.0200,0.0198)(0.0400,0.0390)(0.0600,0.0576)(0.0800,0.0758)(0.1000,0.0935)(0.1200,0.1107)(0.1400,0.1274)(0.1600,0.1436)(0.1800,0.1592)(0.2000,0.1744)(0.2200,0.1892)(0.2400,0.2033)(0.2600,0.2171)(0.2800,0.2306)(0.3000,0.2433)(0.3200,0.2557)(0.3400,0.2678)(0.3600,0.2793)(0.3800,0.2903)(0.4000,0.3009)(0.4200,0.3111)(0.4400,0.3211)(0.4600,0.3306)(0.4800,0.3394)(0.5000,0.3478)(0.5200,0.3560)(0.5400,0.3640)(0.5600,0.3710)(0.5800,0.3780)(0.6000,0.3848)(0.6200,0.3916)(0.6400,0.3982)(0.6600,0.4044)(0.6800,0.4104)(0.7000,0.4163)(0.7200,0.4219)(0.7400,0.4273)(0.7600,0.4325)(0.7800,0.4374)(0.8000,0.4420)(0.8200,0.4465)(0.8400,0.4507)(0.8600,0.4547)(0.8800,0.4585)(0.9000,0.4621)(0.9200,0.4656)(0.9400,0.4688)(0.9600,0.4718)(0.9800,0.4746)(1.0000,0.4773)(1.0100,0.4786)(1.0200,0.4798)(1.0300,0.4810)(1.0400,0.4821)(1.0500,0.4832)(1.0600,0.4843)(1.0700,0.4853)(1.0800,0.4863)(1.0900,0.4873)(1.1000,0.4882)(1.1100,0.4890)(1.1200,0.4899)(1.1300,0.4907)(1.1400,0.4914)(1.1500,0.4921)(1.1600,0.4928)(1.1700,0.4935)(1.1800,0.4941)(1.1900,0.4947)(1.2000,0.4952)(1.2100,0.4958)(1.2200,0.4962)(1.2300,0.4967)(1.2400,0.4971)(1.2500,0.4975)(1.2600,0.4979)(1.2700,0.4982)(1.2800,0.4985)(1.2900,0.4988)(1.3000,0.4990)(1.3100,0.4992)(1.3200,0.4994)(1.3300,0.4996)(1.3400,0.4997)(1.3500,0.4998)(1.3600,0.4999)(1.3700,0.5000)(1.3800,0.5000)(1.3900,0.5000)(1.4000,0.5000)(1.4100,0.5000)(1.4200,0.5000)(1.4300,0.5000)(1.4400,0.5000)(1.4500,0.5000)(1.4600,0.5000)(1.4700,0.5000)(1.4800,0.5000)(1.4900,0.5000)(1.5000,0.5000)};
    \addlegendentry{Inapproximability (\Cref{thm:inapprox-usm-non-pos})}
    
    \node[label={225:{$(0.478, 1)$}},circle,fill,inner sep=1pt,color=red] at (axis cs:1,0.478) {};
    
    \node[label={225:{$(0.5, 1.386)$}},circle,fill,inner sep=1pt,color=red] at (axis cs:1.3862,0.5) {};
    
    \addplot [
        thick,
        color=black,
        fill=black, 
        fill opacity=0.05
    ]
    fill between[
        of=lower and upper,
        soft clip={domain=0:1.4},
    ];

    \end{axis}
    \end{tikzpicture}

    \caption{Graphical presentation of results for \texttt{RegularizedUSM} with a non-positive linear function $\ell$ (\Cref{sec:non-positive}). Following the convention of \cite{bodek2022maximizing}, the $x$ and $y$ axes represent the coefficients of $\ell$ and $f$, respectively. We use blue for approximation algorithms and red for inapproximability results, and the shaded area represents the gap between the best known approximation algorithms and inapproximability results. Observe that \Cref{thm:inapprox-usm-non-pos} unifies the two inapproximability theorems from \cite{bodek2022maximizing}. $(0.5, 2\ln 2-\eps)$-inapproximability is due to \Cref{thm:inapprox-two-ln-two}. For \texttt{RegularizedCSM}, the results are the same for $\beta\le 1$.}
    
    \label{fig:ell-non-pos}
\end{figure}
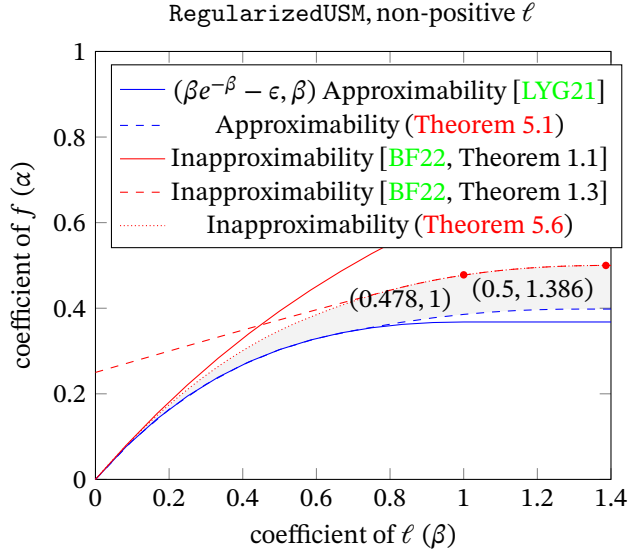

\begin{restatable*}{theorem}{rusmNonposExtended}\label{thm:rusmNonposExtended}

For \RUSM{} with non-positive $\ell$, an $(\alpha(\beta),\beta)$-approximation algorithm exists for any $(\alpha(\beta),\beta)$ in \Cref{tab:0.385-nonpos-extended}. In particular, $\alpha(1)>0.385$ and $\alpha(1.3)>0.398$. When $\beta\le 1$, there is an algorithm for \RCSM{} that achieves the same approximation factor.

\end{restatable*}

A natural follow-up question is whether there is a $(0.5,\beta)$-approximation algorithm for \RUSM{} with non-positive $\ell$ for some $\beta$. Although it is unclear whether this is the case for general $f$, we use linear programming to show this result when $f$ is an undirected or directed cut function (\Cref{thm:max-cut-csm,thm:max-dicut-usm}).

In \Cref{subsec:non-pos-inapprox}, we use the symmetry gap technique to demonstrate improved inapproximability for \RUSM{} with non-positive $\ell$. The previous best inapproximability results were \cite[Theorem 1.1]{bodek2022maximizing} near $\beta=0$ and \cite[Theorem 1.3]{bodek2022maximizing} near $\beta=1$. Our result, which generalizes the construction from \Cref{sec:cardinality-inapprox}, beats or matches both of these theorems for all $\beta$.

\begin{restatable*}{theorem}{inapproxRusmNonpos}\label{thm:inapprox-usm-non-pos}

There are instances of \RUSM{} with non-positive $\ell$ such that $(\alpha(\beta),\beta)$ is inapproximable for any $(\alpha(\beta),\beta)$ in \Cref{tab:inapprox-usm-non-pos}. In particular, $\alpha(0)\approx 0$, matching the result of \cite[Theorem 1.1]{bodek2022maximizing}, and $\alpha(1) < 0.478$, matching the result of \cite[Theorem 1.3]{bodek2022maximizing}.

\end{restatable*}

We conclude this section by showing that taking the limit of \Cref{thm:inapprox-usm-non-pos} as $\alpha(\beta)\to 0.5$ shows $(0.5,2\ln 2-\eps\approx 1.386)$-inapproximability (\Cref{thm:inapprox-two-ln-two}).

\subsubsection*{\Cref{sec:non-neg-usm}: Non-negative $\ell$, \RUSM{}}

The results of this subsection and the next are summarized in \Cref{fig:det-dg,fig:ell-non-neg-usm}.

We note that \Cref{thm:rusmNonposExtended} can be modified to obtain guarantees for \RUSM{} with non-negative $\ell$ (in \Cref{subsec:non-neg-usm-approx}). But first, we take a slight detour and reanalyze the guarantee for this task provided by the randomized double greedy algorithm of \cite{buchbinder2012double} (\RanDG{}), which achieves the best-known $(\alpha(\beta),\beta)$-approximations near $\beta=3/4$. We also reanalyze the guarantee of the deterministic variant of double greedy from the same paper (\DetDG{}).

Recall that \DetDG{} achieves a $1/3$-approximation for \USM{}, while \RanDG{} achieves a $1/2$-approximation for \USM{} in expectation. \cite{bodek2022maximizing} extended these guarantees to \RUSM{} with non-negative $\ell$, showing that \DetDG{} simultaneously achieves $(\alpha,1-\alpha)$-approximations for all $\alpha\in [0,1/3]$, and that \RanDG{} simultaneously achieves $(\alpha,1-\alpha/2)$-approximations for all $\alpha\in [0,1/2]$. In \Cref{subsec:dg-better}, we show improved approximation factors for a variant of \DetDG{} and the original \RanDG{}:

\begin{figure}
    \centering
    \begin{tikzpicture}
    \begin{axis}[
        title={Double Greedy, non-negative $\ell$},
        xlabel={coefficient of $\ell$ ($\beta$)},
        ylabel={coefficient of $f$ ($\alpha$)},
        xmin=0.6, xmax=1,
        ymin=0, ymax=1,
        legend pos=north west,
        % ymajorgrids=true,
        % grid style=dashed,
    ]
    
    \addplot[
        domain=0:1/3,
        samples=2,
        samples y=0,
        color=orange,
    ]({1-x},{x});
    \addlegendentry{\DetDG{}, \cite[Theorem 1.4]{bodek2022maximizing}}
    
    \addplot[
        domain=1:100,
        samples=1000,
        samples y=0,
        color=blue,
    ]({(1+x)/(1+x+1/x)},{(1)/(1+x+1/x)});
    \addlegendentry{Variant of \DetDG{} (\Cref{thm:deterministic-dg-better})}
    
    \addplot[
        domain=0:1/2,
        samples=2,
        samples y=0,
        color=orange,
        dashed,
    ]({1-x/2},{x});
    \addlegendentry{\RanDG{}, \cite[Theorem 1.5]{bodek2022maximizing}}
    
    \addplot[
        domain=1:100,
        samples=1000,
        samples y=0,
        color=blue,
        dashed,
    ]({(2+x)/(2+x+1/x)},{2/(2+x+1/x)});
    \addlegendentry{\RanDG{} (\Cref{thm:randomized-dg-better})}
    
    \end{axis}
    \end{tikzpicture}
    
    \caption{Graphical presentation of improved approximability for both deterministic and randomized double greedy with non-negative $\ell$ (\Cref{subsec:dg-better}).
    }
    
    \label{fig:det-dg}
\end{figure}
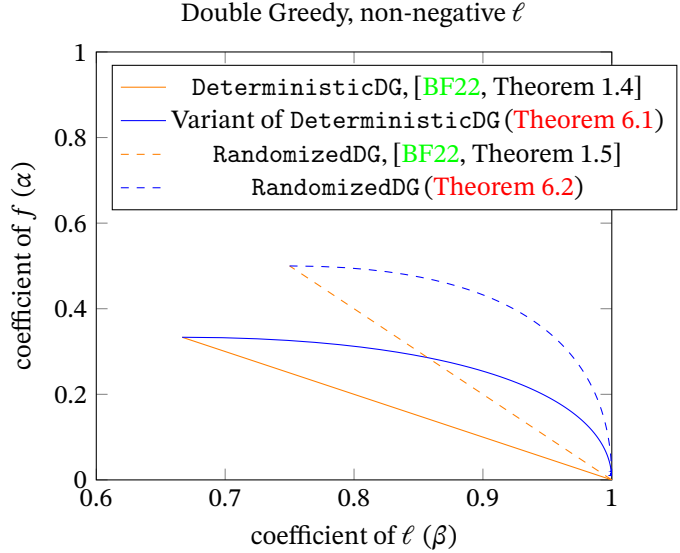

\begin{itemize}
    \item 
    \textbf{Improved analysis of a variant of \DetDG{} (\Cref{thm:deterministic-dg-better}). }
    For any $r\ge 1$, we describe a variant of \DetDG{} that simultaneously achieves $(0,1)$ and $\paren{\frac{1}{r+1+r^{-1}},\frac{r+1}{r+1+r^{-1}}}$-approximations. For $r=1$, the variant is actually just the original \DetDG{}.
    
    \item 
    \textbf{Improved analysis of \RanDG{} (\Cref{thm:randomized-dg-better}).} We then show that \RanDG{} simultaneously achieves $\paren{\frac{2}{r+2+r^{-1}},\frac{r+2}{r+2+r^{-1}}}$-approximations for all $r\ge 1$.
\end{itemize}

Observe that for both \DetDG{} and \RanDG{}, increasing $r$ improves the dependence of the approximation on $\ell$ but decreases the dependence on $f$. Setting $r=1$ recovers the guarantees of \cite{bodek2022maximizing}. We also provide examples showing that neither \DetDG{} nor \RanDG{} achieve $(\alpha,\beta)$-approximations better than \Cref{thm:deterministic-dg-better,thm:randomized-dg-better} in \Cref{thm:deterministic-dg-tight,thm:randomized-dg-tight}, respectively.

In \Cref{subsec:non-neg-usm-approx} we provide improved approximation algorithms for non-negative $\ell$ near $\beta=1$ by combining the results of \Cref{subsec:non-pos-approx,subsec:dg-better}:

\begin{figure}
    \centering
    \begin{tikzpicture}
    \begin{axis}[
        title={\texttt{RegularizedUSM}, non-negative $\ell$},
        xlabel={coefficient of $\ell$ ($\beta$)},
        ylabel={coefficient of $f$ ($\alpha$)},
        xmin=0.75, xmax=1,
        ymin=0, ymax=1,
        legend pos=north west,
        % ymajorgrids=true,
        % grid style=dashed,
    ]
    
    \addplot[
        name path=lower1,
        domain=1:20,
        samples=1000,
        samples y=0,
        color=blue,
    ]({(2+x*x)/(x+1/x)/(x+1/x)},{(2)/(x+1/x)/(x+1/x)});
    \addlegendentry{Approximability (\Cref{thm:randomized-dg-better})}
    
    \node[label={190:{$(0.478,1)$}},circle,fill,inner sep=1pt,color=red] at (axis cs:1,0.478) {};
    \node[circle,fill,inner sep=1pt,color=red] at (axis cs:1,0.4998) {};
    \node[label={110:{$(0.5,0.943)$}},circle,fill,inner sep=1pt,color=red] at (axis cs:0.943, 0.5) {};
    
    \addplot[name path=lower2, color=blue, dashed
    ] coordinates {
        % (0.75, 0.5)
        % (0.76, 0.4996)
        % (0.77, 0.4991)
        % (0.78, 0.4979)
        % (0.79, 0.4964)
        % (0.8, 0.4943)
        % (0.81, 0.4916)
        % (0.82, 0.4885)
        % (0.83, 0.4845)
        (0.84, 0.4799)
        (0.85, 0.4748)
        (0.86, 0.4697)
        (0.87, 0.4646)
        (0.88, 0.4595)
        (0.89, 0.4544)
        (0.9, 0.4492)
        (0.91, 0.4441)
        (0.92, 0.439)
        (0.93, 0.4338)
        (0.94, 0.4283)
        (0.95, 0.4224)
        (0.96, 0.416)
        (0.97, 0.4092)
        (0.98, 0.4018)
        (0.99, 0.3938)
        (1.0, 0.3856)
    };
    \addlegendentry{Approximability (\Cref{thm:rusmNonnegComb})}
    
    % \addplot[
    %     color=black,
    %     ]coordinates {
    %         (0.851808481318819,0.473529993261698)(1,0.385)
    %     };
    % \addlegendentry{Approximability (\Cref{thm:double-greedy-better,thm:0.385-non-neg})}

    \addplot[
        name path=upper,
        domain=0:1, 
        color=red,
    ]{0.5};
    \addlegendentry{Inapproximability of \USM{} \cite{feige2011maximizing}};
    
    \addplot [
        thick,
        color=black,
        fill=black, 
        fill opacity=0.05
    ]
    fill between[
        of=lower1 and upper,
        soft clip={domain=0.75:0.84},
    ];
    
    \addplot [
        thick,
        color=black,
        fill=black, 
        fill opacity=0.05
    ]
    fill between[
        of=lower2 and upper,
        soft clip={domain=0.84:1.0},
    ];
    
    \end{axis}
    \end{tikzpicture}
    
    \caption{Graphical presentation of results for \RUSM{} with a non-negative linear function $\ell$ (\Cref{sec:non-neg-usm}). \cite{feige2011maximizing} showed ($0.5+\eps$)-inapproximability and \cite[Lemma 6.3]{bodek2022maximizing} showed $(0.4998+\eps,1)$-inapproximability. $(0.478,1)$-inapproximability is due to \Cref{thm:0.478-1-inapprox-nonneg}, and $(0.5,2\sqrt 2/3+\eps)$-inapproximability is due to \Cref{thm:0.5-0.943-inapprox-nonneg}.}
    
    \label{fig:ell-non-neg-usm}
\end{figure}
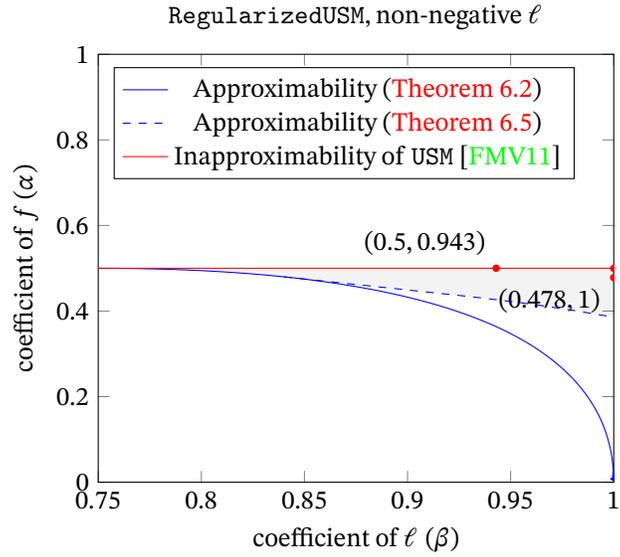

\begin{restatable*}{theorem}{rusmNonnegComb}\label{thm:rusmNonnegComb}

An $(\alpha(\beta),\beta)$-approximation algorithm for \RUSM{} with non-negative $\ell$ exists for any $(\alpha(\beta),\beta)$ in \Cref{tab:nonneg-comb}. In particular, the $\alpha(\beta)$ obtained for $\beta\ge 0.85$ is superior to that of \Cref{thm:randomized-dg-better} alone, and $\alpha(1)>0.385$, matching the approximation factor of \Cref{thm:rusmNonposExtended}.

\end{restatable*}

In \Cref{subsec:non-neg-usm-inapprox}, we use the symmetry gap technique to prove both $(0.478,1-\eps)$ and $(0.5,2\sqrt 2/3\approx 0.943+\eps)$-inapproximability (\Cref{thm:0.478-1-inapprox-nonneg,thm:0.5-0.943-inapprox-nonneg}). These results are much stronger than \cite[Theorem 1.6]{bodek2022maximizing}, which only proved $(0.4998+\eps,1)$-inapproximability. Again, our constructions are variants of that used in \Cref{sec:cardinality-inapprox}.

\subsubsection*{\Cref{sec:non-neg-csm}: Non-negative $\ell$, \RCSM{}} 

The results of this section are summarized in \Cref{fig:ell-non-neg-csm}. To the best of our knowledge, we are the \emph{first} to obtain non-trivial approximations for \RCSM{} when $\ell$ is not necessarily non-positive. In \Cref{subsec:non-neg-csm-approx} we combine the \textit{distorted measured continuous greedy} of \cite{lu2021regularized} with the \textit{aided measured continuous greedy} of \cite{buchbinder2016nonsymmetric} to show the following.

\begin{figure}
    \centering
    
    \begin{tikzpicture}
    \begin{axis}[
        title={\texttt{RegularizedCSM}, non-negative $\ell$},
        xlabel={coefficient of $\ell$ ($\beta$)},
        ylabel={coefficient of $f$ ($\alpha$)},
        xmin=0, xmax=1,
        ymin=0, ymax=1,
        legend pos=north west,
        % ymajorgrids=true,
        % grid style=dashed,
    ]
    \node[label={225:{$\paren{\frac{1}{e},\frac{e-1}{e}}$}},circle,fill,inner sep=1pt, color=blue] at (axis cs:1-1/e,1/e) {};
    
    \addplot[
        color=blue,
        name path=lower,
    ]coordinates{
        (0.0, 0.3857)
(0.01, 0.3857)
(0.02, 0.3857)
(0.03, 0.3857)
(0.04, 0.3857)
(0.05, 0.3857)
(0.06, 0.3857)
(0.07, 0.3857)
(0.08, 0.3857)
(0.09, 0.3857)
(0.1, 0.3857)
(0.11, 0.3857)
(0.12, 0.3857)
(0.13, 0.3857)
(0.14, 0.3857)
(0.15, 0.3857)
(0.16, 0.3857)
(0.17, 0.3857)
(0.18, 0.3857)
(0.19, 0.3857)
(0.2, 0.3857)
(0.21, 0.3857)
(0.22, 0.3857)
(0.23, 0.3857)
(0.24, 0.3857)
(0.25, 0.3857)
(0.26, 0.3857)
(0.27, 0.3857)
(0.28, 0.3857)
(0.29, 0.3857)
(0.3, 0.3857)
(0.31, 0.3857)
(0.32, 0.3857)
(0.33, 0.3857)
(0.34, 0.3857)
(0.35, 0.3857)
(0.36, 0.3857)
(0.37, 0.3857)
(0.38, 0.3857)
(0.39, 0.3857)
(0.4, 0.3857)
(0.41, 0.3857)
(0.42, 0.3857)
(0.43, 0.3857)
(0.44, 0.3857)
(0.45, 0.3857)
(0.46, 0.3857)
(0.47, 0.3857)
(0.48, 0.3857)
(0.49, 0.3857)
(0.5, 0.3857)
(0.51, 0.3857)
(0.52, 0.3857)
(0.53, 0.3857)
(0.54, 0.3857)
(0.55, 0.3857)
(0.56, 0.3857)
(0.57, 0.3857)
(0.58, 0.3857)
(0.59, 0.3857)
(0.6, 0.3848)
(0.61, 0.3822)
(0.62, 0.3774)
(0.63, 0.3699)
(0.64, 0.36)
(0.65, 0.35)
(0.66, 0.34)
(0.67, 0.33)
(0.68, 0.32)
(0.69, 0.31)
(0.7, 0.3)
(0.71, 0.29)
(0.72, 0.28)
(0.73, 0.27)
(0.74, 0.26)
(0.75, 0.25)
(0.76, 0.24)
(0.77, 0.23)
(0.78, 0.22)
(0.79, 0.21)
(0.8, 0.2)
(0.81, 0.19)
(0.82, 0.18)
(0.83, 0.17)
(0.84, 0.16)
(0.85, 0.15)
(0.86, 0.14)
(0.87, 0.13)
(0.88, 0.12)
(0.89, 0.11)
(0.9, 0.1)
(0.91, 0.09)
(0.92, 0.08)
(0.93, 0.07)
(0.94, 0.06)
(0.95, 0.05)
(0.96, 0.04)
(0.97, 0.03)
(0.98, 0.02)
(0.99, 0.01)
(1.0, 0.0)
        };
    \addlegendentry{Approximability (\Cref{thm:ell-nonneg-csm})}

    \addplot[
        domain=0:1,
        color=red,
        name path=upper1,
    ]{0.478};
    \addlegendentry{Inapproximability of \CSM{} \cite{gharan2011submodular}}
    
    \addplot[
        name path=upper2,
        domain=0:1,
        color=red,
        dashed,
    ]{1-x};
    \addlegendentry{Inapproximability (\Cref{thm:inapprox-csm-beta-1})}
    
    \addplot [
        thick,
        color=black,
        fill=black, 
        fill opacity=0.05
    ]
    fill between[
        of=lower and upper1,
        soft clip={domain=0:1-0.478},
    ];
    
    \addplot [
        thick,
        color=black,
        fill=black, 
        fill opacity=0.05,
        color=blue
    ]
    fill between[
        of=lower and upper2,
        soft clip={domain=1-0.478:1},
    ];
    
    \end{axis}
    \end{tikzpicture}

    \caption{Graphical presentation of results for \CSM{} with a non-negative linear function $\ell$ (\Cref{sec:non-neg-csm}).  Recall that \cite[Theorem E.2]{gharan2011submodular} showed $0.478$-inapproximability.
    }
    \label{fig:ell-non-neg-csm}
\end{figure}
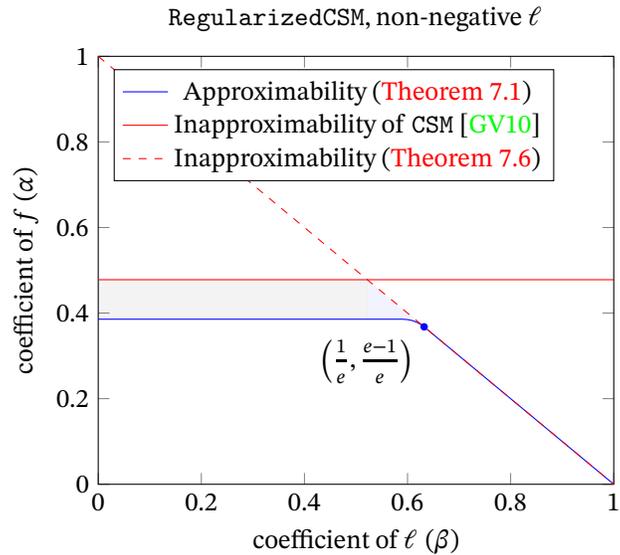

\begin{restatable*}{theorem}{approxNonnegCsm}\label{thm:ell-nonneg-csm}

For \RCSM{} with non-negative $\ell$, there is a $\paren{\alpha(\beta)-\eps,\beta-\eps}$ approximation algorithm for all $\beta\in [0,1]$ where $\alpha$ is a decreasing concave function satisfying $\alpha(0.385)>0.385$, $\alpha(0.6)>0.384, \alpha\paren{1-e^{-1}}=e^{-1}$, and $\alpha(1)=0$.

\end{restatable*}

Note that $\alpha(0.385)>0.385$ matches the (trivial) result of directly applying the algorithm of \cite{buchbinder2016nonsymmetric} to $f+\ell$. In \Cref{subsec:non-neg-csm-inapprox}, we prove a complementary inapproximability result showing that our algorithm is tight for $\beta\ge \frac{e-1}{e}$.

\begin{restatable*}[Inapproximability of \RCSM{} Near $\beta=1$]{theorem}{inapproxNonnegCsm}\label{thm:inapprox-csm-beta-1}

For any $0\le \beta\le 1$, there exist instances of \RCSM{} with non-negative $\ell$ such that a $(1-\beta+\eps,\beta)$-approximation would require exponentially many value queries.

\end{restatable*}

\subsubsection*{\Cref{sec:unconstrained}: Unconstrained $\ell$} 

The results of this section are summarized in \Cref{fig:ell-arbitrary-sign}. In \Cref{subsec:unconstrained-approx} we modify and reanalyze the \textit{distorted measured continuous greedy} of \cite{lu2021regularized} to achieve a better approximation factor for \RUSM{} than \cite[Theorem 1.2]{bodek2022maximizing} for all $\beta \in (0,1)$:

\begin{figure}
    \centering
    \begin{subfigure}[b]{0.45\textwidth}
        \centering
        \begin{tikzpicture}
        \begin{axis}[
            title={\texttt{RegularizedUSM}, unconstrained $\ell$},
            xlabel={coefficient of $\ell$ ($\beta$)},
            ylabel={coefficient of $f$ ($\alpha$)},
            xmin=0, xmax=1,
            ymin=0, ymax=1,
            legend pos=north west,
            width=\linewidth,
            % ymajorgrids=true,
            % grid style=dashed,
        ]
        
        \addplot[
            domain=0:1,
            color=blue,
        ]{x*(1-x)/(1+x)};
        \addlegendentry{Approximability \cite[Theorem 1.2]{bodek2022maximizing}}
        
        \addplot[
            name path=lower,
            domain=0:10,
            samples=1000,
            samples y=0,
            color=blue,
            dashed
        ]({x/(x+exp(-x))}, {x*exp(-x)/(x+exp(-x))});
        \addlegendentry{Approximability (\Cref{thm:ell-arbitrary-sign})}

        \addplot[color=red]coordinates{
        (0.0000,0.0000)(0.0200,0.0198)(0.0400,0.0390)(0.0600,0.0576)(0.0800,0.0758)(0.1000,0.0935)(0.1200,0.1107)(0.1400,0.1274)(0.1600,0.1436)(0.1800,0.1592)(0.2000,0.1744)(0.2200,0.1892)(0.2400,0.2033)(0.2600,0.2171)(0.2800,0.2306)(0.3000,0.2433)(0.3200,0.2557)(0.3400,0.2678)(0.3600,0.2793)(0.3800,0.2903)(0.4000,0.3009)(0.4200,0.3111)(0.4400,0.3211)(0.4600,0.3306)(0.4800,0.3394)(0.5000,0.3478)(0.5200,0.3560)(0.5400,0.3640)(0.5600,0.3710)(0.5800,0.3780)(0.6000,0.3848)(0.6200,0.3916)(0.6400,0.3982)(0.6600,0.4044)(0.6800,0.4104)(0.7000,0.4163)(0.7200,0.4219)(0.7400,0.4273)(0.7600,0.4325)(0.7800,0.4374)(0.8000,0.4420)(0.8200,0.4465)(0.8400,0.4507)(0.8600,0.4547)(0.8800,0.4585)(0.9000,0.4621)(0.9200,0.4656)(0.9400,0.4688)(0.9600,0.4718)(0.9800,0.4746)(1.0000,0.4773)(1.0100,0.4786)(1.0200,0.4798)(1.0300,0.4810)(1.0400,0.4821)(1.0500,0.4832)(1.0600,0.4843)(1.0700,0.4853)(1.0800,0.4863)(1.0900,0.4873)(1.1000,0.4882)(1.1100,0.4890)(1.1200,0.4899)(1.1300,0.4907)(1.1400,0.4914)(1.1500,0.4921)(1.1600,0.4928)(1.1700,0.4935)(1.1800,0.4941)(1.1900,0.4947)(1.2000,0.4952)(1.2100,0.4958)(1.2200,0.4962)(1.2300,0.4967)(1.2400,0.4971)(1.2500,0.4975)(1.2600,0.4979)(1.2700,0.4982)(1.2800,0.4985)(1.2900,0.4988)(1.3000,0.4990)(1.3100,0.4992)(1.3200,0.4994)(1.3300,0.4996)(1.3400,0.4997)(1.3500,0.4998)(1.3600,0.4999)(1.3700,0.5000)(1.3800,0.5000)(1.3900,0.5000)(1.4000,0.5000)(1.4100,0.5000)(1.4200,0.5000)(1.4300,0.5000)(1.4400,0.5000)(1.4500,0.5000)(1.4600,0.5000)(1.4700,0.5000)(1.4800,0.5000)(1.4900,0.5000)(1.5000,0.5000)};
        \addlegendentry{Inapproximability of $\ell\le 0$ (\Cref{thm:inapprox-usm-non-pos})}
            
        \addplot[dashed,name path=upper,color=red]
        coordinates{
        (0.0000,0.0000)(0.0200,0.0198)(0.0400,0.0390)(0.0600,0.0576)(0.0800,0.0758)(0.1000,0.0935)(0.1200,0.1107)(0.1400,0.1274)(0.1600,0.1436)(0.1800,0.1592)(0.2000,0.1744)(0.2200,0.1892)(0.2400,0.2033)(0.2600,0.2171)(0.2800,0.2306)(0.3000,0.2433)(0.3200,0.2557)(0.3400,0.2678)(0.3600,0.2793)(0.3800,0.2903)(0.4000,0.3009)(0.4200,0.3111)(0.4400,0.3211)(0.4600,0.3306)(0.4800,0.3394)(0.5000,0.3478)(0.5200,0.3560)(0.5400,0.3640)(0.5600,0.3710)(0.5800,0.3780)(0.6000,0.3848)(0.6200,0.3909)(0.6400,0.3965)(0.6600,0.4021)(0.6800,0.4070)(0.7000,0.4116)(0.7200,0.4161)(0.7400,0.4199)(0.7600,0.4235)(0.7800,0.4269)(0.8000,0.4295)(0.8200,0.4319)(0.8400,0.4343)(0.8600,0.4364)(0.8800,0.4376)(0.9000,0.4385)(0.9200,0.4393)(0.9400,0.4399)(0.9600,0.4403)(0.9800,0.4399)(1.0000,0.4393)};
        \addlegendentry{Inapproximability (\Cref{thm:inapprox-usm})}

        \node[label={135:{$(0.440,1)$}},circle,fill,inner sep=1pt, color=red] at (axis cs:1,0.439) {};

        \addplot [
            thick,
            color=black,
            fill=black, 
            fill opacity=0.05
        ]
        fill between[
            of=lower and upper,
            soft clip={domain=0:1.4},
        ];
        
        \node[label={215:{$(0.408,1)$}},circle,fill,inner sep=1pt, color=red] at (axis cs:1,0.408) {};
        
        \end{axis}
        \end{tikzpicture}
        
        \caption{\cite[Theorem 1.2]{bodek2022maximizing} provided a $\paren{\frac{\beta(1-\beta)}{1+\beta}-\eps,\beta}$-approximation for \texttt{RegularizedUSM}, and \Cref{thm:inapprox-usm} is stronger than \Cref{thm:inapprox-usm-non-pos} for $\beta$ close to one. $(0.408,1)$-inapproximability is due to \Cref{thm:0.408-inapprox}.}
    \end{subfigure}
    \hfill
    \begin{subfigure}[b]{0.45\textwidth}
        \centering
        \begin{tikzpicture}
        \begin{axis}[
            title={\texttt{RegularizedCSM}, unconstrained $\ell$},
            xlabel={coefficient of $\ell$ ($\beta$)},
            ylabel={coefficient of $f$ ($\alpha$)},
            xmin=0, xmax=1,
            ymin=0, ymax=1,
            legend pos=north west,
            width=\linewidth,
            % ymajorgrids=true,
            % grid style=dashed,
        ]
        
        \addplot[
            name path=lower,
            domain=0:2,
            samples=1000,
            samples y=0,
            color=blue
        ]({x/(x+exp(-x))*(x<1) + (e/(e+1)+(x-1)/(e+1))*(x>1)}, {x*exp(-x)/(x+exp(-x))*(x<1)+(2-x)/(e+1)*(x>1)});
        \addlegendentry{Approximability (\Cref{thm:ell-arbitrary-sign})}

        \addplot[ 
            color=red,
            name path=upper2,
            domain=0:1,
        ]{1-x};
        \addlegendentry{Inapproximability of $\ell\ge 0$ (\Cref{thm:inapprox-csm-beta-1})}

        \addplot[
            dashed,
            color=red,
            name path=upper1,
        ]coordinates{
            (0.0000,0.0000)(0.0200,0.0198)(0.0400,0.0390)(0.0600,0.0576)(0.0800,0.0758)(0.1000,0.0935)(0.1200,0.1107)(0.1400,0.1274)(0.1600,0.1436)(0.1800,0.1592)(0.2000,0.1744)(0.2200,0.1892)(0.2400,0.2033)(0.2600,0.2171)(0.2800,0.2306)(0.3000,0.2433)(0.3200,0.2557)(0.3400,0.2678)(0.3600,0.2793)(0.3800,0.2903)(0.4000,0.3009)(0.4200,0.3111)(0.4400,0.3211)(0.4600,0.3306)(0.4800,0.3394)(0.5000,0.3478)(0.5200,0.3560)(0.5400,0.3640)(0.5600,0.3710)(0.5800,0.3780)(0.6000,0.3848)(0.6200,0.3909)(0.6400,0.3965)(0.6600,0.4021)(0.6800,0.4070)(0.7000,0.4116)(0.7200,0.4161)(0.7400,0.4199)(0.7600,0.4235)(0.7800,0.4269)(0.8000,0.4295)(0.8200,0.4319)(0.8400,0.4343)(0.8600,0.4364)(0.8800,0.4376)(0.9000,0.4385)(0.9200,0.4393)(0.9400,0.4399)(0.9600,0.4403)(0.9800,0.4399)(1.0000,0.4393)};
        \addlegendentry{Inapproximability of \texttt{RUSM} (\Cref{thm:inapprox-usm})}

        \addplot [
            thick,
            color=black,
            fill=black, 
            fill opacity=0.05
        ]
        fill between[
            of=lower and upper1,
            soft clip={domain=0:0.61},
        ];
        
        \addplot [
            thick,
            color=black,
            fill=black, 
            fill opacity=0.05
        ]
        fill between[
            of=lower and upper2,
            soft clip={domain=0.61:1},
        ];
        f
        \node[label={45:{$(0.280,0.7)$}},circle,fill,inner sep=1pt,color=blue] at (axis cs:0.7,0.280) {};
        
        \end{axis}
        \end{tikzpicture}
        \caption{\Cref{thm:little-bit-better} achieves $\alpha=0.280$ for $\beta=0.7$, which is slightly higher than the $\alpha=0.277$ achieved by \Cref{thm:ell-arbitrary-sign}.}
        % \Cref{thm:inapprox-csm-beta-small} in turn is stronger than \Cref{thm:inapprox-usm} for $\beta\ge 0.6$, though the difference is too small to be visible in the graph.
    \end{subfigure}

    \caption{Graphical presentation of results with an unconstrained linear function $\ell$ (\Cref{sec:unconstrained}).}
    \label{fig:ell-arbitrary-sign}
\end{figure}

\begin{restatable*}{theorem}{approxArbitrary}\label{thm:ell-arbitrary-sign}

For all $t\ge 0$, there is a $\paren{\frac{t e^{-t}}{t+e^{-t}}-\eps,\frac{t}{t+e^{-t}}}$-approximation algorithm for \RUSM{}. This algorithm achieves the same approximation guarantee for \RCSM{} when $t\le 1$. 

\end{restatable*}

Note that unlike \cite[Theorem 1.2]{bodek2022maximizing}, our algorithm also applies to \RCSM{}, and is tight for \RCSM{} when $\beta\ge \frac{e}{e+1}$. To the best of our knowledge, this is the first algorithm to achieve any $(\alpha,\beta)$-approximation for \RCSM{} when the sign of $\ell$ is unconstrained. We then demonstrate that our algorithm is not tight for $\beta<\frac{e}{e+1}$; in particular, by combining the methods for non-positive $\ell$ and non-negative $\ell$ (\Cref{sec:non-positive,sec:non-neg-csm}) we achieve a slightly greater value of $\alpha$ for $\beta=0.7$. Note that \Cref{thm:ell-arbitrary-sign} only guarantees a $(0.277, 0.7)$-approximation when $t\approx 0.925$.

\begin{restatable*}{theorem}{arbitraryBitBetter}\label{thm:little-bit-better}

There is a $(0.280, 0.7)$-approximation algorithm for \RCSM{}.

\end{restatable*}

In \Cref{subsec:unconstrained-inapprox} we extend the symmetry gap construction for non-positive $\ell$ from \Cref{thm:inapprox-usm-non-pos} in order to obtain stronger inapproximability results for unconstrained $\ell$. We first show that a natural generalization of \Cref{thm:inapprox-usm-non-pos} proves $(\alpha(\beta),\beta)$-inapproximability where $\alpha(1)<0.440$, and then provide a different construction that shows $(0.408,1)$-inapproximability (\Cref{thm:inapprox-usm,thm:0.408-inapprox}).

\section{Preliminaries}\label{sec:prelims}

We use much the same notation as \cite[Section 2]{bodek2022maximizing}.

\paragraph*{Set Functions.} Let $\NN\triangleq \{u_1,u_2,\dots,u_n\}$ denote the \textit{ground set}. A set function $f:2^{\NN} \to \RR$ is said to be \textit{submodular} if for every two sets $S, T \subseteq \NN$, $f(S) + f(T) \ge f(S \cup T) + f(S \cap T)$. Equivalently, $f$ is said to be submodular if it satisfies the property of ``diminishing returns.'' That is, for every two sets $S \subseteq T \subseteq \NN$ and element $u\in \NN\backslash T$, $f(u|S)\ge f(u|T)$, where $f(u|S)\triangleq f(S\cup \{u\})-f(S)$ is the \textit{marginal value} of $u$ with respect to $S$. We use $f(u)$ as shorthand for $f(\{u\})$. All submodular functions are implicitly assumed to be non-negative unless otherwise stated. 

A set function $f$ is said to be \textit{monotone} if for every two sets $S \subseteq T \subseteq \NN$, $f(S) \le f(T)$, and it is
said to be \textit{linear} if there exist values $\{\ell_u \in \RR | u \in \NN \}$ such that for every set $S \subseteq \NN$, $f(S) = \sum_{u\in S}\ell_u$. When considering the sum of a non-negative submodular function $f$ and a linear function $\ell$ whose sign is unconstrained, define $\NN^+\triangleq \{u\mid u\in \NN\text{ and }\ell(u)> 0\}$ and $\NN^-\triangleq \NN\backslash \NN^+$. In other words, $\NN^+$ contains the elements of the ground set $\NN$ with positive sign in $\ell$ and $\NN^-$ contains all the rest. We additionally define $\ell_+(S)\triangleq \ell(S\cap \NN^+)$ and $\ell_-(S)\triangleq \ell(S\cap \NN^-)$ to be the components of $\ell$ with positive and negative sign, respectively. % One can verify that any linear set function is submodular. 

\paragraph*{Multilinear Extensions.} All vectors of reals are in bold (e.g., $\bx$). Given two vectors $\bx, \by \in [0, 1]^{\NN}$, we define $\bx\vee \by$, $\bx\wedge \by$ and $\bx\circ \by$ to be the coordinate-wise maximum, minimum, and multiplication, respectively,
of $\bx$ and $\by$. We also define $\bx\backslash \by\triangleq \bx - \bx \wedge \by$.

Given a set function $f : 2^{\NN} \to \RR$, its \textit{multilinear extension} is the function $F : [0, 1]^{\NN} \to \RR$ defined by $F(\bx) = \EE[f(\texttt{R}(\bx))]$, where $\texttt{R}(\bx)$ is a random subset of $\NN$ including every element $u\in \NN$ with probability $\bx_u$, independently. One can verify that $F$ is a multilinear function of its arguments as well an extension of $f$ in the sense that $F(\bone_S)=f(S)$ for every set $S\subseteq \NN$. Here, $\bone_S$ is the vector with value 1 at each $u\in S$ and 0 at each $u\in \NN\backslash S$, and is known as the \textit{characteristic vector} of the set $S$. % \textcolor{red}{Many submodular maximization approximation algorithms first approximately maximize the multilinear extension, and then randomly round the resulting fractional solution into an integral solution that preserves the value of the multilinear extension in expectation.}

\paragraph*{Value Oracles.} We make the standard assumption that an algorithm for $f+\ell$ sums accesses $f$ only through a \textit{value oracle}. Given a set $S\subseteq \NN$, a value oracle for $f$ returns $f(S)$ in polynomial time. On the other hand, $\ell$ is directly provided to the algorithm.

\paragraph*{Matroid Polytopes.} A \textit{matroid} $\MM$ may be specified by a pair of a ground set $\NN$ and a family of independent sets $\II$. The \textit{matroid polytope} $\PP$ corresponding to $\MM$ is defined to be $conv(\{\bone_S \mid S\in \II\})$, where $conv$ denotes the convex hull. By construction, $\PP$ is guaranteed to be \textit{down-closed}; that is, $0\le \bx \le \by$ and $\by\in \PP$ imply $\bx \in \PP$. We also make the standard assumption that $\PP$ is \textit{solvable}; that is, linear functions can be maximized over $\PP$ in polynomial time. For \CSM{} and \RCSM{}, we let $OPT$ denote any set such that $OPT\in \II$ (equivalently, $\bone_{OPT}\in \PP$), while for \USM{} and \RUSM{}, we let $OPT$ denote any subset of $\NN$. For example, in the context of \CSM{}, $\EE[f(T)]\ge \alpha f(OPT)$ is equivalent to $\forall S\in \II, \EE[f(T)]\ge \alpha f(S)$.

\paragraph*{Miscellaneous.} We let $\eps$ denote any positive real. Many of our algorithms are ``almost'' $(\alpha,\beta)$ approximations in the sense that they provide an $(\alpha-\eps,\beta)$-approximation in $poly\paren{n,\frac{1}{\eps}}$ time for any $\eps>0$. Similarly, some of our results show $(\alpha+\eps,\beta)$-inapproximability for any $\eps>0$.

\paragraph*{Prior Work.}

A more comprehensive overview than \Cref{sec:intro} of all relevant prior approximation algorithms, together with their corresponding inapproximability results, is deferred to \Cref{subsec:prior-work}.

\section{Inapproximability of Maximization with Cardinality Constraint}\label{sec:cardinality-inapprox}

In this section, we prove \Cref{thm:0.478-inapprox}:

\begin{restatable}{theorem}{inapproxCard}\label{thm:0.478-inapprox}

There exist instances of the problem $\max\{f(S): S\subseteq \NN \text{ and } |S|\le w\}$ such that a 0.478-\allowbreak approximation would require exponentially many value queries.

\end{restatable}

First, we provide the relevant definitions about proving inapproximability using the symmetry gap technique from \cite{vondrak2011symmetry}.

\begin{definition}[Symmetrization]
Let $\GG$ be a group of permutations over $\NN$. For $\bx \in [0, 1]^{\NN}$, define the ``symmetrization of $\bx$'' as $\overline{\bx}=\EE_{\sigma \in \GG}[\sigma(\bx)],$ where $\sigma\in \GG$ is uniformly random and $\sigma(\bx)$ denotes $\bx$ with coordinates permuted by $\sigma$.
\end{definition}

\begin{definition}[Symmetry Gap]\label{def:strongly-symmetric-1}
Let $\max\{f(S): S\in \FF\subseteq 2^{\NN}\}$ be \textit{strongly symmetric} with respect to a group $\GG$ of permutations over $\NN$, meaning that for all $\sigma\in \GG$ and $S\subseteq 2^{\NN}$, $f(S)=f(\sigma(S))$ and $S\in \FF \Leftrightarrow S'\in \FF$ whenever $\overline{\bone_S}=\overline{\bone_{S'}}$. Define $P(\FF)=conv(\{\bone_I : I\in \FF\})$ to be the polytope associated with $\FF$. Then the \textit{symmetry gap} of $\max_\{f(S): S\in \FF\}$ is defined as
\begin{equation*}
\gamma\triangleq \frac{\overline{\BOPT}}{\BOPT}\triangleq \frac{\max_{\bx \in P(\FF)}F(\overline{\bx})}{\max_{\bx\in P(\FF)}F(\bx)}.
\end{equation*}
\end{definition}
% \TS{$\overline{\bx}$ is undefined.}

\begin{lemma}[Inapproximability due to Symmetry Gap]\label{lemma:symmetry-gap-inapprox}

Let $\max \{f(S) : S\in \FF\}$ be an instance of non-negative (optionally monotone) submodular maximization, strongly symmetric with respect to $\GG$, with symmetry gap $\gamma$. Let $\CC$ be the class of instances $\max_\{\tilde f(S) : S \in \tilde \FF\}$ where $\tilde f$ is non-negative submodular and $\tilde \FF$ is a \textit{refinement} of $\FF$. Then for every $\eps > 0$, any (even randomized) $(1+\eps)\gamma$-approximation algorithm for the class $\CC$ would require exponentially many queries to the value oracle for $\tilde f(S)$.

\end{lemma}

The formal definition of refinement can be found in \cite{vondrak2011symmetry}. The important thing to note is that $\tilde{\FF}$ satisfies the same properties as $\FF$. In particular, $\tilde{\FF}$ preserves cardinality and matroid independence constraints. Before proving \Cref{thm:0.478-inapprox}, we start with a related lemma.

\begin{lemma}[Inapproximability of Cardinality Constraint on Subset of Domain]\label{lemma:cardinality-on-subset}

Let $T$ be some subset of the ground set. There exist instances of the problem $\max\{f(S) : S\subseteq \mathcal N \wedge |S\cap T|\le w\}$ such that a 0.478-approximation would require exponentially many value queries.

\end{lemma}

\begin{proof} % [Proof of \Cref{lemma:cardinality-on-subset}]
It suffices to provide $f$ and $\FF$ satisfying the definitions of \Cref{lemma:symmetry-gap-inapprox} with symmetry gap $\gamma<0.478$. The construction is identical to that of \cite[Theorem E.2]{gharan2011submodular}, except we define 
\begin{equation*}
\FF\triangleq\{S\mid S\subseteq \NN \wedge |S \cap \{a_{1\ldots k},b_{1\ldots k}\}|\le 1\}
\end{equation*}
rather than
\begin{equation}
\FF_{orig}\triangleq\{S\mid S\subseteq \NN \wedge |S\cap \{a,b\}|\le 1 \wedge |S \cap \{a_{1\ldots k},b_{1\ldots k}\}|\le 1\}.\label{def:matroid-polytope-orig}
\end{equation}
That is, we drop the constraint $|S\cap \{a,b\}|\le 1$. Recall that \cite[Theorem E.2]{gharan2011submodular} defines the submodular function $f$ as the sum of the weighted cut functions of two directed hyperedges and an undirected edge (see \cite[Figure 4]{gharan2011submodular} for an illustration). Specifically, the weighted cut function on the directed hyperedge $(\{a_1,a_2,\dots,a_k\},a)$ contributes $\kappa\triangleq 0.3513$ to the value of $f(S)$ if $S\cap \{a_1,\dots,a_k\} \neq \emptyset$ and $a\not\in S$, and $0$ otherwise. The weighted cut function on the directed hyperedge $(\{b_1,b_2,\dots,b_k\},b)$ is defined in the same way. Finally, the weighted cut function on the undirected edge $(a,b)$ contributes $1-\kappa$ if $|S\cap \{a,b\}|=1$ and $0$ otherwise. Thus, the multilinear extension of $f$ is as follows:
\begin{align*}
F(\bx_a,\bx_b,\bx_{a_{1\ldots k}},\bx_{b_{1\ldots k}})\triangleq &(1-\kappa)(\bx_a(1-\bx_b)+\bx_b(1-\bx_a))\\
&+\kappa\left[\left(1-\prod_{i=1}^k(1-\bx_{a_i})\right)(1-\bx_a)+\left(1-\prod_{i=1}^k(1-\bx_{b_i})\right)(1-\bx_b)\right].\\
\end{align*}
As in \cite[Lemma 5.4]{gharan2011submodular}, we let $\GG$ be the group of permutations generated by $\{\sigma_1,\sigma_2\}$, where 
\[
\sigma_1(a)=b, \sigma_1(b)=a, \sigma_1(a_i)=b_i, \sigma_1(b_i)=a_i
\]
swaps the two hyperedges and
\[
\sigma_2(a)=a, \sigma_2(b)=b, \sigma_2(a_i)=a_{i\pmod{k}+1}, \sigma_2(b_i)=b_i
\]
rotates the tail vertices of the first hyperedge. It is easy to check that $(f,\FF)$ are strongly symmetric with respect to both $\sigma_1$, and $\sigma_2$, and that the symmetrization of $\bx$ is as follows:
\[
\overline{\bx}=\EE_{\sigma\in \GG}\brac{\sigma(\bx)}=\begin{cases}
\overline{\bx}_a=\overline{\bx}_b=\frac{\bx_a+\bx_b}{2}\\
\overline{\bx}_{a_1}=\dots=\overline{\bx}_{a_k}=\overline{\bx}_{b_1}=\dots=\overline{\bx}_{b_k}=\frac{\sum_{i=1}^{k}\paren{\bx_{a_i}+\bx_{b_i}}}{2k}.
\end{cases}
\]
Observe that 
\[\BOPT\ge \max_{S\in \FF}f(S)\ge f(\{a,b_1\})=(1-\kappa)+\kappa=1.\]
Defining $q\triangleq\frac{\bx_a+\bx_b}{2}$ and $p\triangleq\frac{\sum_{i=1}^k(\bx_{a_i}+\bx_{b_i})}{2}$, the maximum of $F$ over all symmetric $\bx$ is thus:
\begin{align*}
\overline{\BOPT}&=\max_{\bx\in P(\FF)}F(\overline{\bx})\\
&=\max_{\bx\in P(\FF)}F(q,q,p/k,p/k,\ldots,p/k)\\
&\triangleq \max_{\bx\in P(\FF)}\hat F(q,p)\\
&=(1-\kappa)2q(1-q)+\kappa 2(1-q)(1-(1-p/k)^k) \\
&\approx (1-\kappa)2q(1-q)+\kappa 2(1-q)(1-e^{-p})
\end{align*}
where the approximate equality holds as $k\to\infty$. Now,
\[\overline{\BOPT}=\max_{\bx\in P(\FF)}F(\overline \bx)=\max_{p\le 1/2}\hat{F}(q,p)=\max_{p,q\le 1/2}\hat{F}(q,p)=\max_{\bx \in P(\FF_{orig})}F(\overline \bx)<0.478.\]
The third equality holds since $\hat{F}(q,p)\le \hat{F}(1-q,p)$ for $q\in(1/2,1]$ (thus, adding the constraint $q\le 1/2$ has no effect), while the inequality holds due to the proof of \cite[Theorem E.2]{gharan2011submodular}. So the symmetry gap is less than 0.478, as desired.

\end{proof}

Now, all we need to do to show \Cref{thm:0.478-inapprox} is convert the cardinality constraint on $T$ in \Cref{lemma:cardinality-on-subset} into a cardinality constraint on all of $\NN$.

\begin{proof}[Proof of \Cref{thm:0.478-inapprox}]
Again, it suffices to provide $f$ and $\FF$ satisfying the definitions of \Cref{lemma:symmetry-gap-inapprox} with symmetry gap $\gamma<0.478$. We start with the construction from \Cref{lemma:cardinality-on-subset}, replace each element $a_i$ and $b_i$ with $t$ copies $a_{i,1}\ldots a_{i,t}$, and $b_{i,1}\ldots b_{i,t}$ and set $w\triangleq t+1$. The goal is to show that symmetry gap of $f$ with respect to $\FF$ remains less than $0.478$ as $t\to\infty$. Specifically, we may redefine $f$ such that $F$ is as follows:
\begin{align*}
F(\bx_a,\bx_b,\bx_{a_{1\ldots k,1\ldots t}},\bx_{b_{1\ldots k,1\ldots t}})\triangleq &(1-\kappa)(\bx_a(1-\bx_b)+\bx_b(1-\bx_a))\\
&+\kappa\left[\left(1-\prod_{i=1}^k\left(1-\frac{\sum_{j=1}^t\bx_{a_{i,j}}}{t}\right)\right)(1-\bx_a)+\left(1-\prod_{i=1}^k\left(1-\frac{\sum_{j=1}^t\bx_{b_{i,j}}}{t}\right)\right)(1-\bx_b)\right],\\
\end{align*}
Importantly, $f$ remains non-negative submodular and symmetric, with the new symmetrization being as follows for an appropriate choice of $\GG$:
\[
\overline{\bx}=\EE_{\sigma\in \GG}\brac{\sigma(\bx)}=\begin{cases}
\overline{\bx}_a=\overline{\bx}_b=\frac{\bx_a+\bx_b}{2}\\
\overline{\bx}_{a_{1,1}}=\dots=\overline{\bx}_{a_{k,t}}=\overline{\bx}_{b_{1,1}}=\dots=\overline{\bx}_{b_{k,t}}=\frac{\sum_{i=1}^{k}\sum_{j=1}^t\paren{\bx_{a_{i,j}}+\bx_{b_{i,j}}}}{2kt}.
\end{cases}
\]
It can be verified that $F(\overline{\bx})$ can be written in terms of the same function of two variables $\hat{F}(q,p)$ from \Cref{lemma:cardinality-on-subset}:
$$p\triangleq\frac{\sum_{i=1}^k\sum_{j=1}^t(\bx_{a_{ij}}+\bx_{b_{ij}})}{2t}$$
\begin{align}
F(\overline \bx)&\triangleq F(q,q,p/kt,p/kt,\ldots,p/kt)\nonumber\\
&\triangleq \hat F(q,p)\nonumber\\
&=(1-\kappa)2q(1-q)+\kappa 2(1-q)(1-(1-p/k)^k) \label{eq:expect-prod}\\
&\approx (1-\kappa)2q(1-q)+\kappa 2(1-q)(1-e^{-p}) \label{eq:exp}
\end{align}
where \Cref{eq:expect-prod} is satisfied due to the expectation of the product of two independent variables equaling the product of their expectations, and \Cref{eq:exp} holds as $k\to\infty$. As in the proof of \Cref{lemma:cardinality-on-subset},
\[\BOPT\ge \max_{S: |S|\le t+1}f(S)\ge f(\{a,b_{11},\ldots,b_{1t}\})=1,\]
\[\overline{\BOPT}=\max_{\sum \bx_i\le t+1}F(\overline{\bx})=\max_{\sum \bx_i\le t+1}\hat{F}(q,p)\le \max_{p\le \frac{t+1}{2t}}\hat{F}(q,p)\approx \max_{p\le 1/2}\hat{F}(q,p)<0.478,\]
where the approximate equality holds for sufficiently large $t$ because $\lim_{t\to\infty} \frac{t+1}{2t}=\frac{1}{2}$.
\end{proof}

\section{Non-Positive \texorpdfstring{$\ell$}{ell}}\label{sec:non-positive}

The results of this section are summarized in \Cref{fig:ell-non-pos}. 

\subsection{Approximation Algorithms}\label{subsec:non-pos-approx}

% \subsection{A (0.385, 1) Approximation For \texttt{RegularizedCSM}}

In this subsection we provide improved approximations for general $f$ (\Cref{thm:rusmNonposExtended}) as well as for $f$ a cut function  (\Cref{thm:max-cut-csm,thm:max-dicut-usm}). 

\rusmNonposExtended

\begin{table}[h]
    \centering
    \begin{tabular}{|c|c|c|c|c|c|}
    \hline
    $\beta$ & $\beta e^{-\beta}$ & $\alpha(\beta)$\\
    \hline
    0.7 & 0.3476 & 0.3478 \\
    0.8 & 0.3595 & 0.3630 \\
    0.9 & 0.3659 & 0.3757 \\
    1.0 & 0.3679 & 0.3856 \\
    1.1 & 0.3662 & 0.3925 \\
    1.2 & 0.3614 & 0.3967 \\
    1.3 & 0.3543 & 0.3982 \\
    1.4 & 0.3452 & 0.3982 \\
    \hline
    \end{tabular}
    \caption{$(\alpha(\beta),\beta)$-Approximations for \texttt{RegularizedUSM} (\Cref{thm:rusmNonposExtended}). For comparison, the previous best-known approximation factors of \cite{lu2021regularized} are included in the second column.}
    \label{tab:0.385-nonpos-extended}
\end{table}

We start with a special case of \Cref{thm:rusmNonposExtended}.

\begin{lemma}\label{thm:0.385-non-pos}

There is a $(0.385,1)$ approximation algorithm for \texttt{RegularizedCSM} when $\ell$ is non-positive.

\end{lemma}

\begin{proof}
The idea is to combine the ``guessing step'' of Sviridenko et al. \cite{sviridenko2017optimal} with the $0.385$-approximation for \CSM{} due to Buchbinder and Feldman \cite{buchbinder2016nonsymmetric} (which actually provides a $(0.385+\eps)$-approximation for any $\eps\le 0.0006$). Recall that \cite{sviridenko2017optimal} achieves a $\paren{1-\frac{1}{e}-\eps,1}$-approximation for monotone $f$ and non-positive $\ell$. The idea is that if we know the value of $\ell(OPT)$, we can run \cite{buchbinder2016nonsymmetric} on the intersection $\PP\cap \{\bx : L(\bx)\ge \ell(OPT)\}$, which is down-closed and solvable because $\PP$ is down-closed and solvable, and the same is true for $\{\bx : L(\bx)\ge \ell(OPT)\}$. This will guarantee finding $\bx\in \PP$ such that $\EE[F(\bx)]\ge 0.385f(OPT)$ and $\EE[L(\bx)]\ge \ell(OPT)$.

Of course, we do not actually know what the value of $\ell(OPT)$ is. To guarantee that we run \cite{buchbinder2016nonsymmetric} on the intersection $\PP\cap \{\bx : L(\bx)\ge w\}$ for some $w$ satisfying $\ell(OPT)\ge w\ge \ell(OPT)(1+\eps)$, it suffices to try setting $w$ equal to each of the $\BIGO{\frac{n^2}{\eps}}$ values in the following set:
\begin{equation*}
\{0\}\cup \left\{\ell(u)\cdot k\eps \middle| u\in \NN, k\in \ZZ\text{ and }k\in \left[\ceil{\frac{1}{\eps}}, \ceil{\frac{n}{\eps}}\right]\right\}.
\end{equation*}
For at least one of these values of $w$ (``guesses''), we will have $\EE[F(\bx)]\ge (0.385+\eps)f(OPT)$ and $\EE[L(\bx)]\ge (1+\eps)\ell(OPT)$. Combining these guarantees shows that $\bx$ is a $\paren{0.385+\eps,1+\eps}$ approximation, which in turn implies a $\paren{0.385,1}$ approximation  since \[\max(0,(0.385+\eps)f(OPT)+(1+\eps)\ell(OPT))\ge \frac{(0.385+\eps )f(OPT)+(1+\eps)\ell(OPT)}{1+\eps}.\]

\end{proof}

Before proving \Cref{thm:rusmNonposExtended}, we start by briefly reviewing the main algorithm from \cite{buchbinder2016nonsymmetric} when executed on a solvable down-closed polytope $\PP$. First, it uses a \textit{local search} to generate $\bz\in \PP$ such that both of the following inequalities hold with high probability:
\begin{align}
F(\bz)&\ge \frac{1}{2}F(\bz \wedge \bone_{OPT})+\frac{1}{2}F(\bz \vee \bone_{OPT})-o(1)\cdot f(OPT), \label{align:local-1}\\
F(\bz)&\ge F(\bz \wedge \bone_{OPT})-o(1)\cdot f(OPT). \label{align:local-2}
\end{align}
Then it executes \cite[Algorithm 2]{buchbinder2016nonsymmetric}, \textit{Aided Measured Continuous Greedy}, to generate $\by\in \PP$ such that
\begin{align}
\EE[F(\by)]\ge e^{t_s-1}\cdot [(2-t_s-e^{-t_s}-o(1))\cdot f(OPT)&-(1-e^{-t_s})\cdot F(\bz\wedge \bone_{OPT})\nonumber\\
&-(2-t_s-2e^{-t_s})\cdot F(\bz\vee \bone_{OPT})].\label{align:big-expr}
\end{align}
Finally, assuming $\PP$ is the matroid polytope corresponding to a family of independent sets $\II$, the algorithm uses pipage rounding to convert both $\by$ and $\bz$ to integral solutions $y\in \II$ and $z\in \II$ such that $\EE[f(y)]\ge F(\by)$ and $\EE[f(z)]\ge F(\bz)$, and returns the solution from $y$ and $z$ with the larger value of $f$.\footnote{Actually, \cite{buchbinder2014submodular} analyze their algorithm assuming $z$ is returned with probability $p=0.23$ and otherwise $y$, but this distinction is of little consequence.} To obtain improved approximation bounds, we need the following generalization of \Cref{align:big-expr}:

\begin{lemma}[Generalization of \textit{Aided Measured Continuous Greedy}]\label{lemma:aided-ext}
If we run \textit{Aided Measured Continuous Greedy} given 
a fractional solution $\bz$ and a polytope $\PP$ for a total of $t_f$ time, where $t_f\ge t_s$, it will generate $\by\in t_f\PP\cap [0,1]^{\NN}$ such that
\begin{align*}
\EE[F(\by)]\ge e^{-t_f}[(e^{t_s}+t_fe^{t_s}-t_se^{t_s}-1-o(1))f(OPT)
&+(-e^{t_s}+1)F(\bz\wedge \bone_{OPT})\\
+(-e^{t_s}-t_fe^{t_s}+t_se^{t_s}+1+t_f)F(\bz\vee \bone_{OPT})]
\end{align*}
Note that this matches term by term with \Cref{align:big-expr} when $t_f=1$.
\end{lemma}

\begin{proof}[Proof Sketch]
By \cite{buchbinder2016nonsymmetric}, proving the conclusion for integral sets $Z$ implies the conclusion for fractional $\bz$. So it suffices to prove the following.
\begin{align}
\EE[F(\by(t_f))]\ge e^{-t_f}[(e^{t_s}+t_fe^{t_s}-t_se^{t_s}-1-o(1))f(OPT)&+(-e^{t_s}+1)f(OPT\cap Z) \nonumber\\
+(-e^{t_s}-t_fe^{t_s}+t_se^{t_s}+1+t_f)f(OPT\cup Z)].\label{align:aided-ext-integral}
\end{align}
The idea of the original \textit{aided measured continuous greedy} is to run \textit{measured continuous greedy} for $t_s$ time only on the elements of $\NN\backslash Z$, and then for $1-t_s$ additional time with all elements of $\NN$. Working out what happens when we run it for a total of $t_f$ instead of $1$ time is just a matter of going through the equations from \cite[Section 4]{buchbinder2016nonsymmetric} and making a few minor changes. The remainder of the proof is deferred to \Cref{subsec:omitted-proofs}.
\end{proof}

\begin{proof}[Proof of \Cref{thm:rusmNonposExtended}]
Our algorithm for \RUSM{} is as follows:
\begin{enumerate}
    \item As in \Cref{thm:0.385-non-pos}, first guess the value of $\ell(OPT)$ to within a factor of $1+\eps$, and then replace $\PP$ with $\PP\cap \{\bx: L(\bx)\ge (1+\eps)\ell(OPT)\}$.
    
    \item Generate $\bz$ using the local search procedure on $(f,\PP)$ described by \cite{buchbinder2016nonsymmetric}. 
    
    \item Run aided continuous greedy given $\bz$ for all pairs 
    \[(t_f,t_s)\in \TT\triangleq \left\{\paren{\frac{x}{20}, \frac{y}{20}}\middle | (x,y)\in \ZZ^2\text{ and }0\le x\le y\le 40\right\}.\]
    For improved approximations, larger sets $\TT$ can be chosen (but we found the benefit of doing so to be negligible).
    
    \item Round $\bz$ from step 1 and all fractional solutions found in step 2 to valid integral solutions. Note that by replacing $\bz$ with $\texttt{R}(\bz)$, the value of $F+L$ is preserved in expectation.
    
    \item Return the solution from step 4 with the maximum value, or the empty set if none of these solutions has positive expected value. Let $\BOPT'$ be the expected value of this solution.
\end{enumerate}
For a fixed $\beta\ge 0$, we can compute the maximum $\alpha(\beta)$ such that an $(\alpha(\beta)-\BIGO{\eps},\beta)$-approximation is guaranteed by solving the following linear program:
\begin{align*}
\max\,       & x_1 \\
\text{s.t.}\,& (x_1,x_2,x_3,x_4)\in \begin{aligned}[t]
 conv(\{(&0,0,0,0),(0,0.5,0.5,1),(0,1,0,1)\}\,\cup \\ 
         \{(&e^{t_s-t_f} + t_fe^{t_s-t_f} - t_se^{t_s-t_f} - e^{-t_f} - e^{t_s-t_f} + e^{-t_f}, \\
            &- e^{t_s-t_f} - t_fe^{t_s-t_f} + t_se^{t_s-t_f} + e^{-t_f} + t_fe^{-t_f},\\
            &t_f) \mid (t_s,t_f)\in \TT\})
    \end{aligned}\\
\text{and } & x_2\ge 0, x_3\ge 0, x_4\le \beta.\\
\end{align*}
Any point $(x_1,x_2,x_3,x_4)$ within the convex hull satisfies:
\begin{align*}
\BOPT'\ge x_1 f(OPT)+x_2F(\bz\wedge \bone_{OPT})+x_3F(\bz \vee \bone_{OPT})+x_4 L(OPT).
\end{align*}
if we ignore the $o(1)$ terms contributed by \Cref{lemma:aided-ext,align:local-1,align:local-2} and take the limit as $\eps\to 0$. The points determining the hull are as follows:
\begin{itemize}
    \item $(0,0,0,0)$ corresponds to returning the empty set.
    \item $(0,0.5,0.5,1)$ corresponds to $\bz$ satisfying \Cref{align:local-1}.
    \item $(0,1,0,1)$ corresponds to $\bz$ satisfying \Cref{align:local-2}.
    \item The remaining vertices of the hull correspond to running \Cref{lemma:aided-ext} on $\PP$ given $\bz$ for all $(t_s,t_f)\in \TT$.
\end{itemize}
Adding the constraints $x_2,x_3\ge 0$ and $x_4\le \beta$ ensures that
$\BOPT'\ge x_1F(\bz)+\beta L(OPT).$ The results of solving this program with \texttt{CVXPY} \cite{diamond2016cvxpy} for $\beta\in [0,1.5]$ are displayed in \Cref{fig:ell-non-pos}. In particular, $\alpha(1)\ge 0.385$ and the maximum value of $\alpha$ is obtained around $\alpha(1.3)\ge 0.398$.

For the case of \RCSM{}, the reasoning is almost the same, but to ensure that all points returned by \Cref{lemma:aided-ext} lie within $\PP$, we only include pairs in $\TT$ with $t_f\le 1$ in step 3, and pipage rounding with respect to the \textit{original} $\PP$ (not $\PP\cap \{\bx: L(\bx)\ge (1+\eps)\ell(OPT)\}$, which is not necessarily a matroid polytope) must be used for step 4. The results turn out to be identical to those displayed in \Cref{fig:ell-non-pos} for $\beta\le 1$.

\end{proof}

Next, we state better approximation results for $f$ an undirected and directed cut function, respectively. The proofs, which use linear programming, are deferred to \Cref{subsec:omitted-proofs}. We note that linear programming was previously used to provide a 0.5-approximation for \texttt{MAX-DICUT} by Trevisan \cite{trevisan1998parallel} and later by Halperin and Zwick \cite{halperin2001combinatorial}.

\begin{theorem}\label{thm:max-cut-csm}

There is a $(0.5,1)$-approximation algorithm for \RCSM{} when $\ell$ has arbitrary sign and $f$ is the cut function of a weighted undirected graph $(V,E,w)$; that is, for all $S\subseteq V$,
\[f(S)\triangleq\sum_{ab\in E}w_{ab}\cdot [|S\cap\{a,b\}|=1],\] 
where each edge weight $w_{ab}$ is non-negative.

\end{theorem}

Note that while our above result for undirected cut functions applies to \RCSM{}, our subsequent result for directed cut functions only applies to \RUSM{}.

\begin{theorem}\label{thm:max-dicut-usm}

There is a $(0.5,1)$-approximation algorithm for \RUSM{} when $\ell$ has arbitrary sign and $f$ is the cut function of a weighted directed graph $(V,E,w)$; that is, for all $S\subseteq V$,
\[f(S)=\sum_{ab\in \CE}w_{ab}\cdot [a\in S\text{ and }b\not\in S],\] 
where each edge weight $w_{ab}$ is non-negative.

\end{theorem}

\subsection{Inapproximability}\label{subsec:non-pos-inapprox}

In this subsection, we prove \Cref{thm:inapprox-usm-non-pos}. Recall from \Cref{fig:ell-non-pos} that it unifies the guarantees of \cite[Theorem 1.1]{bodek2022maximizing} and \cite[Theorem 1.3]{bodek2022maximizing}. 

\inapproxRusmNonpos

\begin{table}[h]
    \centering
    \begin{tabular}{|c|c|c|c|c|c|}
    \hline
    $\beta$ & $\alpha(\beta)$ \cite[Theorem 1.3]{bodek2022maximizing} & $\alpha(\beta)$ (\Cref{thm:inapprox-usm-non-pos}) & $\kappa$ & $\ell_p$ & $\ell_q$ \\
    \hline
    0.1 & 0.2750 & 0.0935 & 0.6705 & -0.6095 & -0.2680 \\
    0.2 & 0.2998 & 0.1743 & 0.6513 & -0.5322 & -0.2192 \\
    0.3 & 0.3245 & 0.2433 & 0.6498 & -0.4705 & -0.1505 \\
    0.4 & 0.3488 & 0.3008 & 0.6506 & -0.4207 & -0.0893 \\
    0.5 & 0.3728 & 0.3477 & 0.6484 & -0.3800 & -0.0410 \\
    0.6 & 0.3964 & 0.3846 & 0.6388 & -0.3400 & 0.0000 \\
    0.7 & 0.4195 & 0.4162 & 0.5811 & -0.2887 & 0.0000 \\
    0.8 & 0.4420 & 0.4420 & 0.5088 & -0.2289 & 0.0000 \\
    0.9 & 0.4621 & 0.4621 & 0.4335 & -0.1766 & 0.0000 \\
    1.0 & 0.4773 & 0.4773 & 0.3515 & -0.1294 & 0.0000 \\
    \hline
    \end{tabular}
    \caption{Inapproximability of $(\alpha(\beta),\beta)$-approximations for \texttt{RegularizedUSM} with non-positive $\ell$. The parameters $\kappa, \ell_p, \ell_q$ are as described in the proof of \Cref{thm:inapprox-usm-non-pos}.}
    \label{tab:inapprox-usm-non-pos}
\end{table}

Before proving \Cref{thm:inapprox-usm-non-pos}, we state a generalization of the symmetry gap technique to $f+\ell$ sums that we use for \Cref{thm:inapprox-usm-non-pos} and the rest of our inapproximability results. 

\begin{definition}

We say that $\max_{S\in \FF}\brac{f(S)+\ell(S)}$ is \textit{strongly symmetric} with respect to a group of permutations $\GG$ if $\ell(S)=\ell(\sigma(S))$ for all $\sigma\in \GG$ and $(f,\FF)$ are strongly symmetric with respect to $\GG$ as defined in \Cref{def:strongly-symmetric-1}.

\end{definition}

\begin{lemma}[Inapproximability of $(\alpha,\beta)$ Approximations]\label{lemma:symmetry-gap-inapprox-linear}

Let $\max_{S\in \FF}\brac{f(S)+\ell(S)}$ be an instance of non-negative submodular maximization, strongly symmetric with respect to a group of permutations $\GG$. For any two constants $\alpha,\beta\ge 0$, if
\[\max_{\bx \in P(\FF)}\brac{F(\overline \bx)+L(\overline \bx)}<\max_{S\in \FF}[\alpha f(S)+\beta \ell(S)],\]
then no polynomial-time algorithm for \RCSM{} can guarantee a $(\alpha,\beta)$-approximation. The same inapproximability holds for \RUSM{} by setting $\FF=2^{\NN}$.

\end{lemma}

\begin{proof}[Proof Sketch] % [Proof of \Cref{lemma:symmetry-gap-inapprox-linear}]

\cite[Theorem 3.1]{bodek2022maximizing} shows this lemma for the case of \texttt{RegularizedUSM}. The proof for \texttt{RegularizedCSM} is similar, so it is omitted.

\end{proof}

The idea behind the proof of \Cref{thm:inapprox-usm-non-pos} is to generalize the symmetry gap construction of \cite[Theorem 1.3]{bodek2022maximizing}, which in turn is a modification of the 0.478-inapproximability result of \cite{gharan2011submodular} used in \Cref{sec:cardinality-inapprox}.

\begin{proof}[Proof of \Cref{thm:inapprox-usm-non-pos}]

Set $f$ to be the same as defined in \Cref{lemma:cardinality-on-subset}. Now apply \Cref{lemma:symmetry-gap-inapprox-linear} with $S=\{a,b_1\}$. For a fixed $\beta$, we can show $(\alpha,\beta)$-inapproximability using this method if it is possible to choose $\ell$ and $\kappa$ such that the following inequality is true:
\[\max_{\bx \in P(\FF)}\brac{F(\overline \bx)+L(\overline \bx)}<\alpha f(\{a,b_1\})+\beta \ell(\{a,b_1\})=\alpha + \beta \ell(\{a,b_1\})\]
\[\implies \max_{\bx \in P(\FF)}\brac{F(\overline \bx)+L(\overline \bx)}-\beta \ell(\{a,b_1\})<\alpha.\]
So our goal is now to minimize to the LHS of the above inequality. \cite[Theorem 1.3]{bodek2022maximizing} sets $\ell_a=\ell_b=0$, and then chooses $\kappa$ and $\ell_{a_{1\dots k}}=\ell_{b_{1\dots k}}\triangleq \ell_p$ in order to minimize the quantity
\begin{align*}
\max_{\bx \in [0,1]^{\NN}}\brac{F(\overline \bx)+L(\overline \bx)}-\beta \ell(\{a,b_1\})&=\max_{\bx \in [0,1]^{\NN}}\brac{F(\overline \bx)+L(\overline \bx)}-\beta \ell_p\\
&=\max_{0\le q\le 1, 0\le p}\brac{(1-\kappa)2q(1-q)+\kappa 2 (1-q)(1-e^{-p})+2p\ell_p}-\beta \ell_p
\end{align*}
However, choosing $\ell_a=\ell_b\triangleq \ell_q$ to be negative rather than zero gives superior bounds for small $\beta$. That is, our goal is to compute
\begin{equation}
\min_{0\le \kappa\le 1, \ell_q\le 0, \ell_p\le 0}\brac{\max_{0\le q\le 1, 0\le p}\brac{(1-\kappa)2q(1-q)+\kappa 2(1-q)(1-e^{-p})+2p\ell_p+2q\ell_q}-\beta (\ell_p+\ell_q)}.\label{eq:regularized-usm-nonpos}
\end{equation}
We can approximate the optimal value by brute forcing over a range of $(\kappa,\ell_q,\ell_p)$. For $\beta\in \{0.8,0.9,1.0\}$, it is optimal to set $\ell_q=0$, and our guarantee is the same as that of \cite[Theorem 1.3]{bodek2022maximizing}. Our results for $\beta\in \{0.6,0.7\}$ are stronger than those of \cite[Theorem 1.3]{bodek2022maximizing} even though they also satisfy $\ell_q=0$, because that theorem actually only considers $\ell_p\ge -0.5$ and $\kappa\le 0.5$.

\end{proof}

Next, we consider the limit of \Cref{thm:inapprox-usm-non-pos} as $\alpha(\beta)\to 0.5$. Note that this is not a new result in the sense that \cite[Theorem 1.3]{bodek2022maximizing} can already prove it when the parameters $\ell_p$ and $\kappa$ are chosen appropriately, but we nevertheless believe that there is value in explicitly stating it.

\begin{restatable}{theorem}{inapproxRusmNonposLimit}\label{thm:inapprox-two-ln-two}

For any $\eps>0$, there are instances of \RUSM{} with non-positive $\ell$ such that $(0.5,2\ln 2-\eps\approx 1.386)$ is inapproximable.

\end{restatable}

\begin{proof} % [Proof of \Cref{thm:inapprox-two-ln-two}]

To find the maximum $\beta$ such that we can show $(0.5,\beta)$-inapproximability using the construction of \Cref{thm:inapprox-usm-non-pos}, our goal is to choose $\kappa\in (0,0.5)$ and $\ell_p<0$ such that the RHS of the following inequality is maximized:
\begin{equation}
\beta<\frac{0.5-\max_{0\le q\le 1, 0\le p} \brac{(1-\kappa)2q(1-q)+\kappa 2(1-q)(1-e^{-p})+2p\ell_p}}{-\ell_p}\label{ineq:2-ln-2}
\end{equation}
We can rewrite half the expression within the $\max$ as
\begin{align*}
\max_{q, 0\le p}&\brac{(1-\kappa)q(1-q)+\kappa (1-q)(1-e^{-p})+p\ell_p}\\
&= \max_{q,0\le p}\brac{-q^2(1-\kappa)+q(1-\kappa-\kappa (1-e^{-p}))+\kappa(1-e^{-p})+p\ell_p}\\
&=\max_{0\le p}\brac{\frac{(1-2\kappa +\kappa e^{-p})^2}{4(1-\kappa)}+\kappa(1-e^{-p})+p\ell_p},
\end{align*}
so the RHS of \Cref{ineq:2-ln-2} becomes:
\begin{equation}
\frac{2\cdot \min_{0\le p}\brac{\frac{1-\kappa-(1-2\kappa+\kappa e^{-p})^2}{4(1-\kappa)}-\kappa (1-e^{-p})-p\ell_p}}{-\ell_p}=\frac{2\cdot \min_{0\le p}\brac{\frac{\kappa (2e^{-p}-1) - \kappa^2e^{-2p}}{4(1-\kappa)}-p\ell_p}}{-\ell_p}. \label{ineq:2-ln-2-again}
\end{equation}
Next, we claim that for any $p^*>0$, it is possible to choose $\ell_p<0$ such that the numerator of \Cref{ineq:2-ln-2-again} reaches its minimum at $p=p^*$. Define the function $h(p)\triangleq \frac{\kappa (2e^{-p}-1) - \kappa^2e^{-2p}}{4(1-\kappa)}$. It suffices to check that $h$ is decreasing at $p=0$ and concave up for $p\ge 0$; that is, $\frac{d}{dp}h(p)\Bigr|_{p=0}<0$ and $\frac{d^2}{dp^2}\brac{\frac{\kappa (2e^{-p}-1) - \kappa^2e^{-2p}}{4(1-\kappa)}}>0$ for all $p\ge 0$. Both of these inequalities follow from the assumption $\kappa\in (0,0.5)$. 

Finally, when $p^*<\ln 2$, $2e^{-p^*}>1$, implying that $h(p^*)>0$ when $\kappa$ is sufficiently close to 0. For such $p^*$, the RHS of \Cref{ineq:2-ln-2-again} becomes
\begin{equation*}
\frac{2(h(p^*)-p^*\ell_p)}{-\ell_p}\ge \frac{-2p^*\ell_p}{-\ell_p}=2p^*,
\end{equation*}
which can be made arbitrarily close to $\beta<2\ln 2$.

\end{proof}

\section{Non-Negative \texorpdfstring{$\ell$}{ell}: \texttt{RegularizedUSM}}\label{sec:non-neg-usm}

\subsection{Approximations with Double Greedy}\label{subsec:dg-better}

In this subsection, we show improved approximability for \DetDG{} and \RanDG{} in \Cref{thm:deterministic-dg-better,thm:randomized-dg-better}, and then show that both of these results are tight in \Cref{thm:deterministic-dg-tight,thm:randomized-dg-tight}.
The results of this subsection are summarized in \Cref{fig:det-dg}. 

First, we briefly review the behavior of the original \DetDG{} and \RanDG{} of \cite{buchbinder2012double} when executed on a non-negative submodular function $g$, as well as their approximation factors.

\paragraph*{The Algorithm:} Both algorithms construct a sequence of sets $X_i,Y_i$ for $i\in [0,n]$. First, $X_0\triangleq \emptyset$ and $Y_0\triangleq \NN$. Then for each $i$ from $1$ to $n$, execute the following two steps:
\begin{enumerate}
    \item Compute the marginal gains $g(u_i|X_{i-1})=g(X_{i-1}\cup \{u_i\})-g(X_{i-1})$ and $g(u_i|Y_{i-1}\backslash\{u_i\})= g(Y_{i-1})-g(Y_{i-1}\backslash \{u_i\})$. By the original proof of double greedy, 
    \begin{equation}
    g(u_i|X_i)-g(u_i | Y_i\backslash\{u_i\})\ge 0\label{eq:dg-important}
    \end{equation}
    holds by submodularity.
    
    \item Based on the marginal gains, either set $(X_i,Y_i)=(X_{i-1}\cup \{u_i\},Y_{i-1})$ or $(X_i,Y_i)=(X_{i-1},Y_{i-1}\backslash \{u_i\})$.
    \begin{itemize}
        \item In \DetDG{}, the first event occurs if $g(u_i|X_i)\ge -g(u_i | Y_i\backslash\{u_i\})$. 
        
        \item In \RanDG{}, the first event occurs with probability proportional to $a_i\triangleq \max(g(u_i|X_i),0)$, while the second event occurs with probability proportional to $b_i\triangleq \max(-g(u_i | Y_i\backslash\{u_i\}),0)$. In the edge case where $a_i=b_i=0$, it does not matter which event occurs.
    \end{itemize}
\end{enumerate}
Finally, the algorithm returns $X_n=Y_n$. 

\paragraph*{The Approximation Factors:} Let $OPT_i\triangleq (OPT\cup X_i)\cap Y_i$, so that $OPT_0=OPT$ while $OPT_n=X_n=Y_n$. For \DetDG{}, it can be shown via exhaustive casework that:
\begin{equation}
g(OPT_{i-1})-g(OPT_i)\le (g(X_i)-g(X_{i-1}))+(g(Y_i)-g(Y_{i-1})), \label{ineq:deterministic-dg}
\end{equation}
while for \RanDG{}, it can similarly be shown that:
\begin{equation}
\EE[g(OPT_{i-1})-g(OPT_i)]\le \frac{1}{2}\EE\brac{(g(X_i)-g(X_{i-1}))+(g(Y_i)-g(Y_{i-1}))}.\label{ineq:randomized-dg}
\end{equation}
Summing \Cref{ineq:deterministic-dg} from $i=1$ to $i=n$ gives
\begin{align*}
g(OPT)-g(X_n)&\le g(X_n)-g(X_0)+g(Y_n)-g(Y_0)\\
&\le 2g(X_n)-g(\NN)\\
\implies g(X_n)&\ge \frac{g(OPT)+g(\NN)}{3},
\end{align*}
whereas summing \Cref{ineq:randomized-dg} from $i=1$ to $i=n$ gives \[\EE\brac{g(OPT)-g(X_n)}\le \EE\brac{\frac{1}{2}\paren{2g(X_n)-g(\NN)}}\implies \EE\brac{g(X_n)}\ge \frac{2g(OPT)+g(\NN)}{4}.\]
These last two equations imply that if we substitute $f+\ell$ in place of $g$, \DetDG{} and \RanDG{} provide $(1/3,2/3)$- and $(1/2,3/4)$-approximations for \RUSM{}, respectively, because $\ell(OPT)\le \ell(\NN)$. Showing improved $(\alpha,\beta)$-approximations for $\alpha<\frac{1}{3}$ ($\alpha<\frac{1}{2}$) for \DetDG{} (\RanDG{}) is just a matter of modifying \Cref{ineq:deterministic-dg} (\Cref{ineq:randomized-dg}).

\begin{restatable}{theorem}{deterministicDgBetter}\label{thm:deterministic-dg-better}

For \RUSM{} with non-negative $\ell$ and any $r\ge 1$, there exists a variant of \DetDG{} that simultaneously achieves $(0,1)$ and $\paren{\frac{1}{r+1+r^{-1}},\frac{r+1}{r+1+r^{-1}}}$-approximations (and consequently, $(\alpha,\beta)$-approximations for all $(\alpha,\beta)$ on the segment connecting these two points as well). For $r=1$, the variant is actually just the original \DetDG{}.

\end{restatable}

\begin{proof} % [Proof of \Cref{thm:deterministic-dg-better}]
Modify step 2 of \DetDG{} so that the first event occurs if $g(u_i|X_{i-1})\ge -rg(u_i|Y_{i-1}\backslash \{u_i\})$. We claim that the following modified version of \Cref{ineq:deterministic-dg} now holds:
\begin{equation}
g(OPT_{i-1})-g(OPT_i)\le r^{-1}(g(X_i)-g(X_{i-1}))+r(g(Y_i)-g(Y_{i-1})).\label{ineq:deterministic-dg-new}
\end{equation}
First we show that \Cref{ineq:deterministic-dg-new} implies an $\paren{\frac{1}{r+1+r^{-1}},\frac{r+1}{r+1+r^{-1}}}$-approximation. Summing it from $i=1$ to $i=n$ gives:
\begin{align*}
g(OPT)-g(X_n)&\le \paren{r^{-1}+r}g(X_n)-rg(\NN)\\
\implies g(X_n)&\ge \frac{g(OPT)}{r+1+r^{-1}}+\frac{r}{r+1+r^{-1}}g(\NN)\\
&\ge \frac{f(OPT)}{r+1+r^{-1}}+\frac{r+1}{r+1+r^{-1}}\cdot \ell(OPT),
\end{align*}
as desired. Now we show \Cref{ineq:deterministic-dg-new}. First, we consider the case $g(u_i|X_{i-1})\ge -rg(u_i|Y_{i-1}\backslash \{u_i\})$. This assumption implies that $Y_i=Y_{i-1}$, so the last part of \Cref{ineq:deterministic-dg-new} drops out.
\begin{enumerate}
    \item If $u_i\in OPT_{i-1}$, then $OPT_i=OPT_{i-1}$, and \Cref{ineq:deterministic-dg-new} reduces to $0 \le g(u_i|X_{i-1})$, which holds by combining \Cref{eq:dg-important} with the assumption.
    \item If $u_i \not\in OPT_{i-1}$, then $OPT_{i}=OPT_{i-1}\cup \{u_i\}$, then \Cref{ineq:deterministic-dg-new} reduces to 
    \begin{equation*}
    -g(u_i|OPT_{i-1})\le r^{-1}g(u_i|X_{i-1}).
    \end{equation*}
    Since $OPT_{i-1}\subseteq Y_{i-1}\backslash \{u_i\}$), the LHS of this inequality is at most $-g(u_i|Y_{i-1}\backslash \{u_i\})$ by submodularity. On the other hand, the RHS of this inequality is at least $-g(u_i|Y_{i-1}\backslash \{u_i\})$ by assumption.
\end{enumerate}
On the other hand, if $g(u_i|X_{i-1})< -rg(u_i|Y_{i-1}\backslash \{u_i\})$, then $X_i=X_{i-1}$, and the first part of \Cref{ineq:deterministic-dg-new} drops out.
\begin{enumerate}
    \item If $u_i\not\in OPT_{i-1}$, then $OPT_i=OPT_{i-1}$, and \Cref{ineq:deterministic-dg-new} reduces to $0 \le -g(u_i|Y_{i-1}\backslash \{u_i\})$, which holds by combining \Cref{eq:dg-important} with the assumption.
    \item If $u_i\in OPT_{i-1}$, then $OPT_{i}=OPT_{i-1}\backslash \{u_i\}$, then \Cref{ineq:deterministic-dg-new} reduces to 
    \begin{equation*}
    g(u_i|OPT_{i})\le -rg(u_i|Y_{i-1}\backslash \{u_i\}).
    \end{equation*}
    Since $X_i\subseteq OPT_{i}$, the LHS of this inequality is at most $g(u_i|X_i)$ by submodularity. On the other hand, the RHS of this inequality is greater than $g(u_i|X_{i-1})$ by assumption.
\end{enumerate}
It remains to show that this algorithm simultaneously achieves a $(0,1)$-approximation. Because $g(u_i|X_i)-g(u_i | Y_i\backslash\{u_i\})\ge 0$, $\max\paren{g(u_i|X_i), -rg(u_i | Y_i\backslash\{u_i\}}\ge 0$. Thus, the values of $g(X_i)$ and $g(Y_i)$ are increasing over the course of the algorithm, so:
\[g(X_n)=g(Y_n)\ge g(Y_{n-1})\ge \dots \ge g(Y_0)\ge \ell(OPT).\qedhere\]
\end{proof}

The reasoning for \RanDG{}, which we show next, is very similar. 

\begin{restatable}{theorem}{randomizedDgBetter}\label{thm:randomized-dg-better}

Running \RanDG{} on $f+\ell$ simultaneously achieves an $\paren{\frac{2}{r+2+r^{-1}},\frac{r+2}{r+2+r^{-1}}}$- approximation for all $r\ge 1$ for \RUSM{} with non-negative $\ell$. 

\end{restatable}

We note that a similar lemma as \Cref{thm:randomized-dg-better} was previously used by \cite{buchbinder2014submodular} for maximizing a submodular function subject to a cardinality constraint.

\begin{proof} % [Proof of \Cref{thm:randomized-dg-better}]
We claim that the following modified version of \Cref{ineq:randomized-dg} holds for any $r>0$:
\begin{equation}
\EE[g(OPT_{i-1})-g(OPT_i)]\le \frac{1}{2}\EE\brac{r^{-1}(g(X_i)-g(X_{i-1}))+r(g(Y_i)-g(Y_{i-1}))}.\label{ineq:randomized-dg-new}
\end{equation}
As in the proof of \Cref{thm:deterministic-dg-better}, it is easy to check that \Cref{ineq:randomized-dg-new} implies the conclusion. It remains to show \Cref{ineq:randomized-dg-new}. We note that in the edge case $a_i=b_i=0$, \Cref{eq:dg-important} implies that $g(u_i|X_{i-1})=g(u_i|Y_{i-1}\backslash \{u_i\})=0$, so the inequality reduces to $0\le 0$. Otherwise, recall that the original proof of double greedy lower bounded the LHS of \Cref{ineq:randomized-dg-new} by
\begin{equation*}
\EE[g(OPT_{i-1})-g(OPT_i)]\le \frac{a_ib_i}{a_i+b_i}.
\end{equation*}
On the other hand, we can lower bound twice the RHS by
\begin{align*}
\EE[&r^{-1}(g(X_i)-g(X_{i-1}))+r(g(Y_i)-g(Y_{i-1}))]\\
&= r^{-1}\cdot \frac{a_i}{a_i+b_i}(g(X_{i-1}\cup\{u_i\})-g(X_{i-1}))+r\cdot \frac{b_i}{a_i+b_i}(g(Y_{i-1}\backslash \{u_i\})-g(Y_{i-1}))\\
&=\frac{r^{-1}a_i^2}{a_i+b_i}+\frac{rb_i^2}{a_i+b_i} \ge \frac{2a_ib_i}{a_i+b_i},
\end{align*}
where the last step follows from the AM-GM inequality as in the original proof.
\end{proof}

Next, we prove that \DetDG{} and \RanDG{} do no better than the bounds we just showed. Recall that \cite[Theorem 1.4]{bodek2022maximizing} proved that the original \DetDG{} is an $(\alpha,\beta)$-approximation algorithm whenever $\alpha\le \frac{1}{3}$ and $\alpha+\beta\le 1$. To show that this analysis is tight, it suffices to check that whenever $\alpha>\frac{1}{3}$ or $\alpha+\beta>1$, there are instances where \DetDG{} does not achieve the desired approximation factor. The former inequality holds by \cite[Theorem II.3]{buchbinder2012double}, while the latter holds by applying the following theorem with $r=1$:

\begin{theorem}\label{thm:deterministic-dg-tight}

For any $r\ge 1$ and $\eps>0$, there are instances of \RUSM{} with non-negative $\ell$ where the variant of \DetDG{} described in the proof of \Cref{thm:deterministic-dg-better} does not achieve an $(\alpha,\beta)$-approximation for any $(\alpha,\beta)$ above the line connecting $(0,1)$ and $\paren{\frac{1}{r+1+r^{-1}}, \frac{r+1}{r+1+r^{-1}}}$.

\end{theorem}

\begin{proof} % [Proof of \Cref{thm:deterministic-dg-tight}]
The points $(\alpha,\beta)$ lying above the line connecting $(0,1)$ and $\paren{\frac{1}{r+1+r^{-1}}, \frac{r+1}{r+1+r^{-1}}}$ are precisely those that satisfy $\alpha+\beta r=r+\eps$ for some $\eps>0$. Define $f(S)$ to be the sum of two weighted cut functions:
\[\NN\triangleq \{u_1,u_2\}\]
\[f(S)\triangleq (r+\eps/2) \cdot [u_1\in S\text{ and }u_2\not\in S]+1\cdot [u_2\in S\text{ and }u_1\not\in S]\]
\[\ell(u_1)=0, \ell(u_2)=r\]
The weights of the directed edges are chosen such that if the variant of \DetDG{} considers $u_1$ before $u_2$, it will compute
\[g(u_1|X_0)=r+\eps/2>-rg(u_1|Y_0\backslash \{u_1\})=r,\]
so it will return a set $T$ satisfying $u_1\in T$, implying that $f(T)+\ell(T)\le r+\eps/2$ regardless of whether $u_2\in T$ or not. If we define $OPT\triangleq \{u_2\}$, then $f(OPT)=1$ and $\ell(OPT)=r$, so we get
\[f(T)+\ell(T)=r+\eps/2 <r+\eps = \alpha +\beta r= \alpha f(OPT) + \beta \ell(OPT)\]
implying that an $(\alpha,\beta)$-approximation is not achieved.
\end{proof}

Next, we generalize the construction of \Cref{thm:deterministic-dg-tight} to show that \Cref{thm:randomized-dg-better} is tight for \RanDG{}.

\begin{theorem}\label{thm:randomized-dg-tight}

For any $r\ge 1$ and $\eps>0$, there are instances of \RUSM{} with non-negative $\ell$ where \RanDG{} does not provide an $\paren{\alpha,\beta}=\paren{\frac{2}{r+2+r^{-1}}+\eps,\frac{r+2}{r+2+r^{-1}}}$-approximation.

\end{theorem}

\begin{proof} % [Proof of \Cref{thm:randomized-dg-tight}]
Define $f(S)$ to be the sum of $2(n-1)$ weighted directed cut functions:
\[
f(S)=\frac{1}{n-1}\brac{\sum_{i=1}^{n-1}\paren{r[u_i\in S\wedge u_n\not\in S]+[u_n\in S\wedge u_i\not\in S]}}
\]
and $\ell(u_1)=\ell(u_2)=\dots=\ell(u_{n-1})=0, \ell(u_n)=r-1$. For each $i\in [1,n-1]$, \RanDG{} will compute $a_i=\frac{r}{n-1}$ and $b_i=\frac{1}{n-1}$, so it will include each of $u_{1\dots n-1}$ in its returned set $X_n$ independently with probability $\frac{r}{r+1}$ each. Thus, for any $\eps>0$, the following inequality holds by a Chernoff bound for sufficiently large $n$:
\[
\Pr\brac{\abs{\frac{|X\cap \{u_1,\dots,u_{n-1}\}|}{n-1}-\frac{r}{r+1}}\ge \frac{\eps}{2r}}=o(1),
\]
Assuming $\abs{\frac{|X_n\cap \{u_1,\dots,u_{n-1}\}|}{n-1}-\frac{r}{r+1}}< \frac{\eps}{2r}$ holds, it follows that
\[
f(X)<r\cdot \paren{\frac{r}{r+1}+\frac{\eps}{2r}}=\frac{r^2}{r+1}+\frac{\eps}{2}.
\]
regardless of whether $u_n$ is included in $X_n$ or not. On the other hand, if we define $OPT\triangleq \{u_n\}$, then
\begin{align*}
\alpha f(OPT)+\beta\ell(OPT)&= \alpha +(r-1)\beta\\
&=\frac{2}{r+2+r^{-1}}+\eps+\frac{(r-1)(r+2)}{r+2+r^{-1}}\\
&=\frac{r^2+r}{r+2+r^{-1}}+\eps\\
&=\frac{r^2}{r+1}+\eps.
\end{align*}
As $f(X_n)<\alpha f(OPT)+\beta \ell(OPT)-\frac{\eps}{2}$ with high probability and $f(X_n)$ is bounded above by a constant independent of $n$, $\EE[f(X_n)]<\alpha f(OPT)+\beta \ell(OPT)$ for sufficiently large $n$, implying that \RanDG{} does not provide an $(\alpha,\beta)$ approximation for this instance.
\end{proof}

Unfortunately, neither version of double greedy achieves any $(\alpha,\beta)$-approximation when $\ell$ is non-positive rather than non-negative. We defer further discussion to \Cref{subsec:online-algos}.

% In \Cref{subsec:online-algos}, we discuss whether the general class of \textit{online} algorithms (of which double greedy is just one example) can achieve approximation guarantees for \RUSM{} when $\ell$ is not necessarily non-negative. We note that neither \DetDG{} nor \RanDG{} provides any approximation guarantee for $\ell$ non-positive, even in the special case of $f$ a directed cut function.

\subsection{Additional Approximation Algorithms}\label{subsec:non-neg-usm-approx}

In this subsection we prove \Cref{thm:rusmNonnegComb}. 
The results of this subsection and the next are summarized in \Cref{fig:ell-non-neg-usm}. 

\rusmNonnegComb

\begin{table}[h]
    \centering
    \begin{tabular}{|c|c|c|c|c|}
    \hline
    $\beta$ & $\alpha(\beta)$ (\Cref{thm:rusmNonnegComb}) & $\alpha(\beta)$ (\Cref{thm:randomized-dg-better}) \\
    \hline
    0.85 & 0.4749 & 0.4746 \\
    0.9 & 0.4493 & 0.4325 \\
    0.95 & 0.4226 & 0.3472 \\
    1 & 0.3856 & 0 \\
    \hline
    \end{tabular}
    \caption{$(\alpha(\beta),\beta)$-approximations for \texttt{RegularizedUSM} with non-negative $\ell$}
    \label{tab:nonneg-comb}
\end{table}

First, we show that the result for $\beta=1$ easily follows from \Cref{thm:0.385-non-pos}.

\begin{restatable}{lemma}{rusmNonnegBetaOne}\label{thm:rusmNonnegBetaOne}

For \RUSM{} with non-negative $\ell$, there is a $(0.385,1)$-approximation algorithm.

\end{restatable}

\begin{proof} % [Proof of \Cref{thm:rusmNonnegBetaOne}]

Define $g(S)\triangleq f(\NN\backslash S)$, which is also non-negative submodular. Then apply \Cref{thm:0.385-non-pos} on $(g,-\ell)$ to find $T\subseteq \NN$ such that
\begin{align*}
\EE[g(T)-\ell(T)]&\ge \max_S[0.385\cdot g(S)-\ell(S)]\\
&=\max_S[0.385 f(S)-\ell(\NN\backslash S)]\\
&=\max_S[0.385 f(S)+\ell(S)]-\ell(\NN).
\end{align*}
Setting $T'=\NN\backslash T$, we have
\begin{align*}
\EE[f(T')-\ell(\NN\backslash T')]&=\EE[f(T')-\ell(\NN)+\ell(T')]\\
&\ge \max_S[0.385 f(S)+\ell(S)]-\ell(\NN).
\end{align*}
Adding $\ell(\NN)$ to both sides, we conclude that
\begin{align*}
\EE[f(T')+\ell(T')]\ge \max_S[0.385 f(S)+\ell(S)].
\end{align*}
So an algorithm returning $T'$ would achieve a $(0.385,1)$-approximation as desired.
\end{proof}

For $\beta$ close to one, we can obtain better $(\alpha,\beta)$-approximations than \Cref{thm:randomized-dg-better} alone provides by combining double greedy with the following corollary of \Cref{thm:rusmNonnegBetaOne}: 

\begin{corollary}\label{corollary:0.385-non-neg}
An $(\alpha,\beta)$-approximation algorithm for \texttt{RegularizedUSM} for the case of $\ell$ non-positive may be used to return a set $T\subseteq \NN$ such that
\begin{equation*}
\EE[f(T)]\ge \alpha f(OPT)+\beta \ell(OPT)+(1-\beta)\ell(\NN).
\end{equation*}
for the case of $\ell$ non-negative.
\end{corollary}

Now we can prove \Cref{thm:rusmNonnegComb} by combining \Cref{corollary:0.385-non-neg} with \Cref{thm:randomized-dg-better}.

\begin{proof}[Proof of \Cref{thm:rusmNonnegComb}]

Our algorithm returns the best of the solutions returned by the following two algorithms:
\begin{enumerate}
    \item Double greedy on $f+\ell$
    \item \Cref{corollary:0.385-non-neg} using \Cref{thm:rusmNonposExtended} for $\beta\in \TT\triangleq\{(\alpha(1+0.01x),1+0.01x)\mid x\in \ZZ\text{ and }0\le x\le 30\}$
\end{enumerate}

As with \Cref{thm:rusmNonposExtended}, for a fixed $\beta$ we can lower bound the $\alpha(\beta)$ guaranteed by the algorithm above by the solution to the following linear program after choosing the set $\mathcal R$ appropriately:
\begin{align*}
\max\,       & x_1 \\
\text{s.t.}\,& (x_1,x_2,x_3)\in \begin{aligned}[t]
conv\Biggl(&\{(\alpha,\beta,1-\beta) \mid (\alpha,\beta)\in \TT\text{ and }\exists (\alpha,\beta)\text{-approximation algorithm for }\ell \le 0\} \cup  \\
           &\left\{\paren{\frac{2}{(r+1/r)^2}, \frac{2}{(r+1/r)^2}, \frac{r^2}{(r+1/r)^2}} \middle|r\in \mathcal R\right\}\Biggr)
\end{aligned}\\
\text{and}\, & x_2+x_3\ge \beta, x_3\ge 0
\end{align*}
Let $\BOPT'$ denote the expected value of the returned solution. Any point $(x_1,x_2,x_3)$ within the convex hull satisfies the following inequality:
$$\BOPT'\ge x_1 f(OPT) + x_2 \ell(OPT)+x_3\ell(\NN).$$
The conditions $x_2+x_3\ge \beta, x_3\ge 0$ ensure that $\BOPT'\ge x_1 f(OPT)+\beta \ell(OPT)$.
\end{proof}

\subsection{Inapproximability} \label{subsec:non-neg-usm-inapprox}

In this subsection, we prove \Cref{thm:0.478-1-inapprox-nonneg,thm:0.5-0.943-inapprox-nonneg}. 

\begin{restatable}{theorem}{inapproxRusmNonneg}\label{thm:0.478-1-inapprox-nonneg}

For some $\eps>0$, there are instances of \RUSM{} with non-negative $\ell$ such that $(0.478,1-\eps)$ is inapproximable.

\end{restatable}

Note that this is much stronger than the $(0.4998+\eps,1)$-inapproximability provided by \cite[Lemma 6.3]{bodek2022maximizing}.

\begin{proof}[Proof Sketch]

We start by showing $(0.478,1)$-inapproximability, which is easier. It suffices to show the following generalization for \Cref{thm:rusmNonnegBetaOne}. Any $(\alpha,1)$-approximation algorithm for the \RUSM{} instance $(f(\NN\backslash S),-\ell(S))$ immediately implies a $(\alpha,1)$-approximation algorithm for $(f(S),\ell(S))$. Letting $\NN\backslash T$ be the set returned by the former approximation algorithm, we find
\begin{align*}
\EE[f(\NN \backslash (\NN\backslash T))-\ell(\NN\backslash T)]&\ge \alpha f(\NN\backslash OPT')-\ell(OPT')\\
\implies \EE\brac{f(T)-\ell(\NN)+\ell(T)} &\ge \alpha f(\NN\backslash OPT')-\ell(OPT')
\end{align*}
Substituting $OPT'=\NN\backslash OPT$ gives
\begin{align*}
\EE\brac{f(T)-\ell(\NN)+\ell(T)}&\ge \alpha f(OPT)-\ell(\NN)+\ell(OPT)\\
\implies \EE[f(T)+\ell(T)]&\ge \alpha f(OPT)\ell(OPT).
\end{align*}
Note that when $\ell$ is set to be non-negative, this means that any $(\alpha,1)$-approximation algorithm for $\ell$ non-positive implies an $(\alpha,1)$-approximation algorithm for $\ell$ non-negative. Similarly, by setting $\ell$ to be non-positive, we get the implication in the opposite direction. This also means that $(\alpha,1)$-inapproximability results for one sign of $\ell$ can be converted to corresponding inapproximability results for the other sign of $\ell$. Thus, the $(0.478,1)$-inapproximability result for non-positive $\ell$ implies the same inapproximability result for non-negative $\ell$.

The slightly stronger result of $(0.478,1-\eps)$ inapproximability for some $\eps>0$ follows from modifying the symmetry gap construction of \Cref{thm:inapprox-usm-non-pos}. Let $(f_{-},\ell_{-})$ be the $f$ and $\ell$ defined in \Cref{thm:inapprox-usm-non-pos} for $\beta=1$. Then let
\[f(S)\triangleq f_{-}(\NN\backslash S), \ell(S)\triangleq -\ell_-(S).\]
For $k$ sufficiently large, this instance shows $(\alpha',1)$-inapproximability for some $\alpha<0.478$. Furthermore, if we fix $k$ to be constant, then the desired result follows; specifically, we can choose $\eps>0$ such that
\[
\alpha' f(OPT) + \ell(OPT) = 0.478 f(OPT)+(1-\eps)\ell(OPT).\qedhere
\]
\end{proof}

Next, we provide an inapproximability result for $\alpha=0.5$ by fixing $k=2$ in the construction for \Cref{thm:0.478-1-inapprox-nonneg}.

\begin{restatable}{theorem}{inapproxRusmNonnegLimit}\label{thm:0.5-0.943-inapprox-nonneg}

For any $\eps>0$, there are instances of \RUSM{} with non-negative $\ell$ such that $(0.5,2\sqrt 2/3 \approx 0.943+\eps)$ is inapproximable.

\end{restatable}

\begin{proof} % [Proof of \Cref{thm:0.5-0.943-inapprox-nonneg}]

Again, let $(f_-,\ell_-)$ be the $f$ and $\ell$ defined in \Cref{thm:inapprox-usm-non-pos}. Define
\[f(S)\triangleq f_-(\NN\backslash S),\]
\[p\triangleq k-\frac{\sum_{i=1}^k(\bx_{a_i}+\bx_{b_i})}{2}\in [0,k],\]
\[q\triangleq 1-\frac{\bx_a+\bx_b}{2}\in [0,1].\]
\[\ell(S)\triangleq -\ell_-(S)=\ell_p(2k-2p)\]
where we may choose any real number $\ell_p>0$. Applying \Cref{lemma:symmetry-gap-inapprox-linear}, we find that the LHS is given by
\[\max_{\bx\in [0,1]^{\NN}}[(F+L)(\overline{\bx})]=\max_{0\le p\le k, 0\le q\le 1}\brac{(1-\kappa)2q(1-q)+\kappa 2(1-q)(1-\paren{1-p/k}^k)-2p\ell_p}+2k\ell_p,\]
while the RHS is bounded below by
\[(\alpha f+\beta \ell)(\NN\backslash \{a,b_1\})=\alpha+\beta\brac{(2k-1)\ell_p}.\]
Now fix $k=2$ and $\alpha=0.5$, and define
\[g(p)\triangleq \max_{0\le q\le 1}\brac{(1-\kappa)2q(1-q)+\kappa 2(1-q)(1-\paren{1-p/k}^k)}.\]
Then the minimum $\beta$ such that we can show  $(\alpha,\beta+\eps)$-inapproximability using this technique is given by
\[\max_{0\le p\le k}\brac{g(p)-2p\ell_p}+2k\ell_p=0.5+\beta^*[(2k-1)\ell_p]\]
\[
\implies \beta^*=\min_{0\le \kappa\le 1, 0<\ell_p}\brac{\frac{\max_{0\le p\le k}\brac{g(p)-2p\ell_p}+2k\ell_p-0.5}{(2k-1)\ell_p}}.
\]
Choose any $p^*\in (0,2-\sqrt 2)$, which automatically guarantees $1-(1-p^*/k)^k<\frac{1}{2}$. For any such $p^*$, we claim that there exist $\kappa$ and $\ell_p$ such that 
\begin{equation}
\max_{0\le p\le k}\brac{g(p)-2p\ell_p}<0.5-2p^*\ell_p. \label{ineq:p-star}
\end{equation}
The reason why \Cref{ineq:p-star} holds is that, for sufficiently small $\kappa>0$, $g(p^*)<0.5$, $g(p)$ is increasing with respect to $p$, and $g(p)$ is concave down with respect to $p$. Thus, we can always choose $\ell_p>0$ so that $\text{argmax}_{0\le p\le k}[g(p)-2p\ell_p]=p^*$. From \Cref{ineq:p-star} we can finish as follows:
\[\beta^*\le \frac{2k\ell_p-2p^*\ell_p}{(2k-1)\ell_p}=\frac{4-2p^*}{3}.\]
Taking the limit as $p^*\to 2-\sqrt 2$ shows the inapproximability of $\beta^*=\frac{4-2(2-\sqrt 2)}{3}+\eps=\frac{2\sqrt 2}{3}+\eps$, as desired.

\end{proof}

\section{Non-Negative \texorpdfstring{$\ell$}{ell}: \texttt{RegularizedCSM}}\label{sec:non-neg-csm}

The results of this section are summarized in \Cref{fig:ell-non-neg-csm}. 

\subsection{Approximation Algorithms}\label{subsec:non-neg-csm-approx}

In this subsection we prove \Cref{thm:ell-nonneg-csm}. 

\approxNonnegCsm

Recall from \Cref{subsubsec:regularized-submod} that \cite{lu2021regularized} introduced \textit{distorted measured continuous greedy} and analyzed its guarantee for the case of non-positive $\ell$. Our improved results are based on generalizing the analysis to the case where $\ell$ contains both positive and negative components.

\begin{lemma}[Generalized Guarantee of Distorted Measured Continuous Greedy]\label{lemma:extended-distorted-measured}

For unconstrained $\ell$ and any $t_f\in [0,1]$, there is a polynomial-time algorithm for \RCSM{} that returns $T\in \II$ such that
\begin{equation*}
\EE[f(T)+\ell(T)]\ge (t_f e^{-t_f}-o(1))f(OPT)+(1-e^{-t_f}) \ell(OPT\cap \NN^+)+t_f\ell(OPT\cap \NN^-)
\end{equation*}
When $t_f>1$, the approximation guarantee still holds, although it is possible that $T\not\in \II$.
\end{lemma}

\begin{proof}[Proof Sketch]

Recall from \Cref{subsec:prior-work} that \cite{lu2021regularized} only prove that their algorithm (Algorithm 1) guarantees an $(e^{-1}-o(1),1)$-approximation for the case of $\ell$ non-positive when run for a total of $t_f=1$ time with a distorted objective of $G_t(\by)=e^{t-1}F(\by)+L(\by)$. Thus we only describe how to modify their analysis for $t_f=1$, but as suggested by \cite{bodek2022maximizing}, this argument can easily be generalized to general $t_f$.

Essentially, the only part of the analysis of \cite[Algorithm 1]{lu2021regularized} that needs to be changed is \cite[Lemma 3.6]{lu2021regularized}, which originally states that when $\ell$ is non-positive,
\begin{equation*}
L(\mathbf{y}(t+\delta))-L(\mathbf{y}(t))=\delta L(\mathbf{z}(t)\circ (\bone_{\NN}-\mathbf{y}(t)))\ge \delta\langle\mathbf{\ell},\mathbf{z}(t)\rangle.
\end{equation*}
Here, time has been discretized into timesteps of size $\delta>0$ where $\delta$ is sufficiently small and $\delta$ evenly divides $t_f$. To generalize this lemma to unconstrained $\ell$, we combine this reasoning with \cite[Lemma 3.1]{lu2021regularized}, which states that $\mathbf{y}_e(t)\le 1-(1-\delta)^{\frac{t}{\delta}}$. It follows that
\begin{align*}
L(\mathbf{y}(t+\delta))-L(\mathbf{y}(t))&=\delta L(\mathbf{z}(t)\circ (\bone_{\NN}-\mathbf{y}(t))) \\
&=\delta\paren{\langle \ell_+,\mathbf{z}(t)\circ (\bone_{\NN}-\mathbf{y}(t))\rangle + \langle \ell_-,\mathbf{z}(t)\circ (\bone_{\NN}-\mathbf{y}(t))\rangle} \\
&\ge \delta\paren{\langle\mathbf{\ell}_-,\mathbf{z}(t)\rangle+\langle\mathbf{\ell}_+,\mathbf{z}(t)\rangle \cdot (1-\delta)^{\frac{t}{\delta}}}.
\end{align*}
To finish, 
\[\sum_{i=0}^{t_f/\delta-1}\delta (1-\delta)^i\ge 1-(1-\delta)^{t_f/\delta}\ge 1-e^{-t_f},\]
giving us the desired coefficient for $\ell(OPT\cap \NN^+)$ when $\bone_{OPT}$ is substituted in place of $\bz(t)$.
\end{proof}

\begin{corollary}\label{corollary:regularizedcsm-nonneg-ell-e}

When $\ell\ge 0$, there is a $\paren{e^{-1}-\eps,1-e^{-1}}$-approximation algorithm for \texttt{RegularizedCSM}.

\end{corollary}

\begin{proof}
The result follows immediately from substituting $\beta=1$ into \Cref{lemma:extended-distorted-measured}.
\end{proof}

In fact, we can obtain the following generalization of \Cref{corollary:regularizedcsm-nonneg-ell-e}, although before proving it we will need two more lemmas. The first lemma is simple.

\begin{lemma}[Trivial Approximation for \texttt{RegularizedCSM}]\label{lemma:trivial-approx}

When $\ell$ is unconstrained, there exists a $(0,1)$-approximation algorithm for \texttt{RegularizedCSM}.

\end{lemma}

\begin{proof}

Because $\PP$ is solvable, we can maximize the linear function $\ell$ over it. Then, because $\PP$ is a matroid independence polytope, pipage rounding can be used to round the fractional solution returned by linear programming to a valid solution within $\PP$ while preserving the value of $\ell$ in expectation.

\end{proof}

The next lemma combines \Cref{lemma:extended-distorted-measured} with the \textit{aided measured continuous greedy} used by \cite{buchbinder2016nonsymmetric}.

\begin{lemma}[Guarantee of \textit{Distorted Aided Measured Continuous Greedy}]\label{lemma:distorted-aided}
Let $\ell$ be unconstrained. If we run \textit{Distorted Aided Measured Continuous Greedy} given a fractional solution $\bz$ and a polytope $\PP$ for a total of $t_f$ time, where $t_f\ge t_s$, it will generate $\by\in t_f\PP\cap \paren{(1-e^{-t_f})\cdot [0,1]^{\NN}}$ such that
\begin{align*}
\EE\brac{F(\by)+L(\by)}\ge e^{-t_f}[&(e^{t_s}+t_fe^{t_s}-t_se^{t_s}-1-o(1))f(OPT)+(-e^{t_s}+1)F(\bz\wedge \bone_{OPT})\\
&\phantom{e^{-t_f}[}+(-e^{t_s}-t_fe^{t_s}+t_se^{t_s}+1+t_f)F(\bz \vee \bone_{OPT})]\\
&+(1-e^{-t_f})L_+(\bone_{OPT}\backslash \bz)+(1-e^{t_s-t_f})L_+(\bone_{OPT}\wedge \bz)\\
&+t_f L_-(\bone_{OPT}\backslash \bz) + (t_f-t_s)L_-(\bone_{OPT}\wedge \bz).
\end{align*}
% \begin{align*}
% F(\by)+L(\by)\ge e^{-t_f}[(e^{t_s}+t_fe^{t_s}-t_se^{t_s}-1-o(1))f(OPT)
% &+(-e^{t_s}+1)f(OPT\cap Z)\\
% +(-e^{t_s}-t_fe^{t_s}+t_se^{t_s}+1+t_f)f(OPT\cup Z)]\\
% +(1-e^{-t_f}-o(1))\ell_+(OPT\backslash Z)+(1-e^{t_s-t_f}-o(1))\ell_+(OPT\cap Z)\\
% +t_f \ell_-(OPT\backslash Z) + (t_f-t_s)\ell_-(OPT\cap Z).
% \end{align*}
Note that the terms depending on $f$ are precisely the same as those in \Cref{lemma:aided-ext}.
\end{lemma}

\begin{proof}

As with \Cref{lemma:aided-ext}, we only present an informal proof assuming direct oracle access to the multilinear extension $F$ and giving the algorithm in the form of a continuous-time algorithm. The techniques mentioned in \cite{buchbinder2016nonsymmetric} and \cite{lu2021regularized} can be used to formalize this at the cost of introducing the $o(1)$ term. 

Let $G(\by(t))\triangleq e^{t-t_f}F(\by(t)))+L(\by(t))$ be the value of the distorted objective at time $t$. Then
\begin{align}
\frac{dG(\by(t))}{dt}\ge & e^{t-t_f}\cdot 
\begin{cases}
f(OPT\backslash Z)-(1-e^{-t})f(OPT\cup Z) & t\in [0,t_s) \\
e^{t_s-t}f(OPT)-(e^{t_s-t}-e^{-t})f(OPT\cup Z) & t\in [t_s,t_f)
\end{cases}\nonumber\\
&+e^{-t}\cdot \begin{cases}
\ell_+(OPT\backslash Z) & t\in [0,t_s)\\
\ell_+(OPT\backslash Z)+e^{t_s}\ell(OPT\wedge Z) & t\in [t_s,t_f)
\end{cases}\nonumber\\
&+\begin{cases} 
\ell_-(OPT\backslash Z) & t\in [0,t_s) \\
\ell_-(OPT) & t\in [t_s,t_f) \\
\end{cases}\label{ineq:aided-ext-ext}
\end{align}
Here, the first and third terms of the summation correspond directly to those of the original aided measured continuous greedy, while the second comes from observing that $\by_u(t)\le 1-e^{-t}$ for $u\in OPT\backslash Z$ and $\by_u(t)\le 1-e^{-\max(t_s-t,0)}$ for $u\in OPT\wedge Z$.

To lower bound $G(\by(t_f))$, we can integrate \Cref{ineq:aided-ext-ext} from $t=0$ to $t=t_f$. As expected, the dependence on $f$ turns out to be the same as \Cref{lemma:aided-ext}.
\end{proof}

\begin{proof}[Proof of \Cref{thm:ell-nonneg-csm}]
The algorithm is similar to that of \Cref{thm:rusmNonposExtended}.

\begin{enumerate}
    \item Run the trivial approximation algorithm (\Cref{lemma:trivial-approx}).
    
    \item Generate $\bz$ using the local search procedure described by \cite[Lemma 3.1]{buchbinder2016nonsymmetric} on $(f+\ell,\PP)$. This finds $\bz\in \PP$ such that
    \begin{align}
    F(\bz)+L(\bz)&\ge \frac{(F+L)(\bz\vee \bone_{OPT})+(F+L)(\bz\wedge \bone_{OPT})}{2}-o(1)\cdot (f+\ell)(OPT)\nonumber\\
    &\ge \frac{1}{2}F(\bz\vee \bone_{OPT})+\frac{1}{2}F(\bz\wedge \bone_{OPT})+\frac{1}{2}\ell(OPT)+\frac{1}{2}L(\bz\wedge \bone_{OPT})-o(1)\cdot (f+\ell)(OPT),\label{align:convex-opt-1}
    \end{align}
    and
    \begin{equation}
    F(\bz)+L(\bz)\ge F(\bz\wedge \bone_{OPT})+L(\bz\wedge \bone_{OPT})-o(1)\cdot (f+\ell)(OPT).\label{align:convex-opt-2}
    \end{equation}
    Note that unlike \Cref{thm:rusmNonposExtended}, there is no guessing step.
    
    \item Run distorted aided measured continuous greedy given $\bz$ (\Cref{lemma:distorted-aided}), for all pairs 
    \[(t_f,t_s)\in \TT\triangleq \{(0.1x,1)\mid 0\le x\le 10\}.\]
    
    \item Round $\bz$ from step 1 and all fractional solutions found in steps 2 and 3 to valid integral solutions using pipage rounding, which preserves the value of $F+L$ in expectation.
    
    \item Return the solution from step 4 with the maximum value, or the empty set if none of these solutions has positive expected value. Let $\BOPT'$ be the expected value of this solution.
\end{enumerate}

As in the proof of \Cref{thm:rusmNonposExtended}, for a fixed $\beta$, we claim that to find a lower bound on $\alpha$ such that the following inequality is true:
\begin{equation*}
\BOPT'\triangleq \max\paren{\EE[F(\bz)+L(\bz)],\max_{(t_s,t_f)\in \TT}\paren{\EE\brac{F(\by_{t_s,t_f})+L(\by_{t_s,t_f})}}}\ge \alpha F(OPT) +\beta \ell(OPT),
\end{equation*}
it suffices to solve the following linear program:
\begin{align*}
\max\,        & x_1 \\
\text{s.t.}\, & (x_1,x_2,x_3,x_4,x_5)\in
\begin{aligned}[t]
conv(\{(&0,0,0,1,1),(0,0.5,0.5,0.5,1),(0,1,0,0,1)\}\cup \\ 
\{(&e^{t_s-t_f} + t_fe^{t_s-t_f} - t_se^{t_s-t_f} - e^{-t_f}, \\
&- e^{t_s-t_f} + e^{-t_f}, \\
&- e^{t_s-t_f} - t_fe^{t_s-t_f} + t_se^{t_s-t_f} + e^{-t_f} + e^{-t_f}t_f,\\
&1-e^{-t_f}, \\
&1-e^{t_s-t_f}) | (t_s,t_f)\in \TT\})\\
\end{aligned} \\
\text{and}\,& x_2\ge 0, x_3\ge 0, x_4\ge \beta, x_5\ge \beta.
\end{align*}
Any point $(x_1,x_2,x_3,x_4,x_5)$ within the convex hull satisfies:
$$\BOPT'\ge x_1f(OPT)+x_2 F(\bz \wedge \bone_{OPT})+x_3 F(\bz \vee \bone_{OPT})+x_4 L(\bone_{OPT}\backslash \bz)+x_5 L(\bz\wedge \bone_{OPT})$$
ignoring the $o(1)$ terms. The points determining the hull are as follows:
\begin{itemize}
    \item $(0,0,0,1,1)$ corresponds to \Cref{lemma:trivial-approx}
    \item $(0,0.5,0.5,0.5,1)$ corresponds to \Cref{align:convex-opt-1}
    \item $(0,1,0,0,1)$ corresponds to \Cref{align:convex-opt-2}
    \item The remaining points correspond to \Cref{lemma:distorted-aided} for all $(t_s,t_f)\in \TT$.
\end{itemize}
The constraints $x_2,x_3\ge 0$ ensure that
$$\BOPT'\ge x_1f(OPT)+x_4 L(\bone_{OPT}\backslash \bz)+x_5 L(\bz\wedge \bone_{OPT}).$$
The constraints $\min(x_4,x_5)\ge \beta$ ensure that
\[\BOPT'\ge x_1 f(OPT)+\beta \ell(OPT).\qedhere\]
\end{proof}

\subsection{Inapproximability}\label{subsec:non-neg-csm-inapprox}

In this subsection, we prove \Cref{thm:inapprox-csm-beta-1}, which can be used to show that \Cref{thm:ell-nonneg-csm} is tight for $\beta\ge (e-1)/e$. We then discuss whether the construction used in \Cref{thm:inapprox-csm-beta-1} could potentially be extended to \RUSM{}.

\inapproxNonnegCsm

\begin{proof} % [Proof of \Cref{thm:inapprox-csm-beta-1}]

Let $\alpha\triangleq 1-\beta+\eps$. By \Cref{lemma:symmetry-gap-inapprox-linear}, it suffices to construct a submodular function $f$ satisfying
\begin{equation}
\max_{\bx\in \PP}[F(\overline{\bx})+L(\overline{\bx})]< \max_{S\in \II}[\alpha\cdot f(S)+\ell(S)],\label{eq:suffices}
\end{equation}
where $\PP$ is the matroid polytope corresponding to a matroid $\MM=(\NN,\II)$. We use the same $f$ that \cite{vondrak2011symmetry} uses for proving the inapproximability of maximization over matroid bases. Specifically, we consider the \textit{Maximum Directed Cut} problem on $k$ disjoint arcs; that is, $f(S)\triangleq \sum_{i=1}^k[a_i\in S\text{ and }b_i\not\in S]$. Its multilinear extension is as follows:
\begin{equation*}
F(\bx_{a_1\dots a_k}, \bx_{b_1 \dots b_k})=\sum_{i=1}^k\bx_{a_i}(1-\bx_{b_i}),
\end{equation*}
We define the independent sets of the matroid to be precisely the subsets of $\NN$ that contain at most one element from $a_1,\dots,a_k$ and at most $k-1$ elements from $b_1,\dots, b_k$, resulting in the following matroid independence polytope:
\begin{equation*}
\PP=\left\{(\bx_{a_i},\bx_{b_i})\middle| \sum_{i=1}^k\bx_{a_i}\le 1\text{ and }\sum_{i=1}^k\bx_{b_i}\le k-1\right\}.
\end{equation*}
Finally, we define $\ell$ as follows:
\begin{equation*}
\ell(a_i)=0, \ell(b_i)=\frac{1}{k}.
\end{equation*}
Then the RHS of \Cref{eq:suffices} is at least:
\begin{equation*}
\max_{S\in \II}[\alpha f(S)+\beta\ell(S)]\ge (\alpha f+\beta \ell)(\{a_1,b_2,b_3,\dots,b_k\})=\alpha + \beta \cdot \frac{k-1}{k},
\end{equation*}
while the LHS of \Cref{eq:suffices} corresponds to the value of the best \textit{symmetrized} solution $\overline{\bx}$, which is $\overline{\bx}_{a_i}=\frac{1}{k}, \overline{\bx}_{b_i}=\frac{k-1}{k}$, giving the following:
\begin{equation*}
\max_{\bx\in \PP}[F(\overline{\bx})+L(\overline{\bx})]=\frac{1}{k}+\frac{k-1}{k}=1.
\end{equation*}
For sufficiently large $k$ we have $\alpha+\beta\cdot \frac{k-1}{k}\ge (\alpha+\beta)\frac{k-1}{k}=(1+\eps)\cdot \frac{k-1}{k}>1$.
\end{proof}

In fact, the bound of \Cref{thm:inapprox-csm-beta-1} is (nearly) tight for $\beta$ close to one.
% \todo{check whether there's anything stronger than this?}

\begin{restatable}[Tight \RCSM{} Near $\beta=1$ for $\ell\ge 0$]{corollary}{tightRcsmBetaOne}\label{corollary:tight-csm-beta-1-ell-nonneg}

For all $\frac{e-1}{e}\le \beta< 1$, there is a $(1-\beta-\eps,\beta)$-approximation algorithm for \RCSM{} with non-negative $\ell$, nearly matching the bound of \Cref{thm:inapprox-csm-beta-1}.

\end{restatable}

\begin{proof} % [Proof of \Cref{corollary:tight-csm-beta-1-ell-nonneg}]

The better of \Cref{corollary:regularizedcsm-nonneg-ell-e} and \Cref{lemma:trivial-approx} will be an $(\alpha,\beta)$-approximation for all $(\alpha,\beta)$ lying above the segment connecting $\paren{\frac{1}{e}-\eps,\frac{e-1}{e}-\eps}$ and $(0,1)$.

\end{proof}

As the $f$ used by \Cref{lemma:symmetry-gap-inapprox-linear} to prove \Cref{thm:inapprox-csm-beta-1} is just a directed cut function, it is natural to ask whether directed cut functions can be used by \Cref{lemma:symmetry-gap-inapprox-linear} to show improved inapproximability for \RUSM{}. We build on \Cref{thm:max-dicut-usm} to show that doing so is impossible.

\begin{theorem}\label{thm:dicut-inapprox}

When $\ell$ is unconstrained, setting $f$ to be a directed cut function in \Cref{lemma:symmetry-gap-inapprox-linear} cannot be used to show $(0.5,1)$-inapproximability for \RUSM{}.

\end{theorem}

The proof is deferred to \Cref{subsec:omitted-proofs}.

\section{Unconstrained \texorpdfstring{$\ell$}{ell}}\label{sec:unconstrained}

The results of this section are summarized in \Cref{fig:ell-arbitrary-sign}. 

\subsection{Approximation Algorithms}\label{subsec:unconstrained-approx}

In this subsection we prove \Cref{thm:ell-arbitrary-sign,thm:little-bit-better}. 

\approxArbitrary

Recall that Bodek and Feldman \cite[Theorem 1.2]{bodek2022maximizing} guaranteed a $\paren{\frac{\beta(1-\beta)}{1+\beta}-\eps,\beta}$ approximation for \RUSM{} using a \textit{local search} technique. \Cref{thm:ell-arbitrary-sign} improves on this approximation factor for all $\beta\in (0,1)$ and also provides guarantees for \RCSM{}.

% First, we show how to combine \Cref{lemma:trivial-approx} with \Cref{lemma:extended-distorted-measured} to show the desired theorem.

\begin{proof} % [Proof of \Cref{thm:ell-arbitrary-sign}]

Our algorithm simply returns the better of the solutions returned by the following two algorithms:
\begin{enumerate}
    \item The set $T$ returned by running \Cref{lemma:extended-distorted-measured} (Distorted Measured Continuous Greedy) for $t_f=t$ time
    \item The set $T'$ returned by \Cref{lemma:trivial-approx} (Trivial Approximation)
\end{enumerate}
Now we show that the desired approximation factor is achieved. Disregard the factors of $o(1)$ in \Cref{lemma:extended-distorted-measured}; they can always be accounted for later at the cost of introducing the factor of $\eps$.
Add $t+e^{-t}-1$ times the inequality of \Cref{lemma:trivial-approx} to the inequality from \Cref{lemma:extended-distorted-measured}.
\begin{align*}
(t+e^{-t})\EE[\max\paren{f(T)+\ell(T),\ell(T')}]&\ge \EE[f(T)+\ell(T)] + (t+e^{-t}-1) \EE[\ell(T')]\\
&\ge t e^{-t}f(OPT)+t (\ell(OPT\cap \NN^+)+\ell(S\cap \NN^-))\\
&=t e^{-t}f(OPT)+t \ell(OPT).
\end{align*}
Then divide both sides by $t+e^{-t}$ and return the set out of $T$ and $T'$ that gives a higher value of $f+\ell$, giving the desired result after accounting for $\eps$:
\begin{equation*}
\EE[\max(f(T)+\ell(T),\ell(T'))]\ge \paren{\frac{t e^{-t}}{t+e^{-t}}-\eps}f(OPT)+\frac{t}{t+e^{-t}} \ell(OPT).\qedhere
\end{equation*}
\end{proof}

Next we show that \Cref{thm:ell-arbitrary-sign} is tight near $\beta=1$.

\begin{restatable}[Tight \RCSM{} Near $\beta=1$]{corollary}{arbitraryTight}\label{corollary:tight-csm-beta-1}

There is a $(1-\beta-\eps,\beta)$-approximation algorithm for \RCSM{} for any $\frac{e}{e+1}\le \beta< 1$, nearly matching the bound of \Cref{thm:inapprox-csm-beta-1}.

\end{restatable}

\begin{proof} % [Proof of \Cref{corollary:tight-csm-beta-1}]

Setting $t=1$, the output of \Cref{lemma:extended-distorted-measured} is both a $\paren{\frac{1}{e+1}-\eps,\frac{e}{e+1}}$-approximation and a $(0,1)$-approximation for \texttt{RegularizedCSM}. Therefore it is also an $(\alpha,\beta)$-approximation for all $(\alpha,\beta)$ lying above the segment connecting $\paren{\frac{1}{e+1}-\eps,\frac{e}{e+1}}$ and $(0,1)$.

\end{proof}

However, our result is not tight for $\beta<e/(e+1)$; it turns out that it is possible to do a little better than \Cref{thm:ell-arbitrary-sign} for $\beta$ near 0.7 by making use of \Cref{lemma:distorted-aided}.

\arbitraryBitBetter

\begin{proof} % [Proof of \Cref{thm:little-bit-better}]
The algorithm is \Cref{thm:ell-nonneg-csm} augmented to use the guessing step from \Cref{thm:rusmNonposExtended}. That is, we start by guessing the value of $\ell_-(OPT)$ to within a factor of $1+\eps$ and replacing $\PP$ with $\PP\cap \{\bx : L_-(\bx)\ge (1+\eps)\ell_-(OPT)\}$ as in \Cref{thm:rusmNonposExtended}, and then run \Cref{thm:ell-nonneg-csm}.

To analyze the guarantee of this algorithm, we set up a linear program similar to that of \Cref{thm:ell-nonneg-csm} with two additional variables $x_6$ and $x_7$ corresponding to $L_-(\bone_{OPT}\backslash \bz)$ and $L_-(\bz \wedge \bone_{OPT})$, respectively. Again, we ignore terms that are $o(1)$ and those depending on $\eps$.
\begin{align*}
\max\, & x_1\\
\text{s.t.}\, & (x_1,x_2,x_3,x_4,x_5,x_6,x_7)\in 
\begin{aligned}[t] % https://tex.stackexchange.com/questions/396371/vertical-top-alignment-of-text-and-inline-math-using-aligned
conv(\{(&0,0,0,1,1,0,0),\\
       (&0,0.5,0.5,0.5,1,1,1),\\
       (&0,1,0,0,1,0,1)\}\cup \\ 
     \{(&e^{t_s-t_f} + t_fe^{t_s-t_f} - t_se^{t_s-t_f} - e^{-t_f} - e^{t_s-t_f} + e^{-t_f}, \\
        & - e^{t_s-t_f} - t_fe^{t_s-t_f} + t_se^{t_s-t_f} + e^{-t_f} + e^{-t_f}t_f,\\
        & 1-e^{-t_f}, 1-e^{t_s-t_f}, t_f, t_f-t_s) \mid (t_s,t_f)\in \TT\})
\end{aligned} \\
\text{and } &x_2\ge 0, x_3\ge 0, x_4\ge \beta, x_5\ge \beta, x_6\le \beta, x_7\le \beta.
\end{align*}
The points determining the hull are as follows:
\begin{itemize}
    \item $(0,0,0,1,1,0,0)$ corresponds to \Cref{lemma:trivial-approx}.
    \item $(0,0.5,0.5,0.5,1,1,1)$ corresponds to \Cref{align:convex-opt-1}. Note that this inequality holds only because of the guessing step.
    \item $(0,1,0,0,1,0,1)$ corresponds to \Cref{align:convex-opt-2}.
    \item The remaining vertices correspond to \Cref{lemma:distorted-aided}.
\end{itemize}
Choosing $\TT=\{(0.205, 0.955)\}$ and solving the linear program gives $x_1\ge 0.280$ as desired.
\end{proof}

We conclude by noting that an analogue of \Cref{corollary:tight-csm-beta-1-ell-nonneg} (Tight \texttt{RegularizedCSM} Near $\beta=1$ for $\ell\ge 0$) holds for unconstrained $\ell$, though for a smaller range of $\beta$.

\subsection{Inapproximability}\label{subsec:unconstrained-inapprox}

In this subsection we prove \Cref{thm:inapprox-usm,thm:0.408-inapprox}. Note that \Cref{thm:inapprox-csm-beta-1} cannot possibly apply to \texttt{RegularizedUSM} because \Cref{thm:ell-arbitrary-sign} achieves $(1-\beta+\eps,\beta)$-approximations for $\beta$ close to one. Unfortunately, we are unable to prove $(1,\eps)$-inapproximability of \texttt{RegularizedUSM}, but we modify \Cref{thm:inapprox-usm-non-pos} to show improved inapproximability for unconstrained $\ell$ than for $\ell$ non-negative or $\ell$ non-positive.

\begin{restatable}[Inapproximability of \RUSM{}]{theorem}{inapproxRusm}
\label{thm:inapprox-usm}

There are instances of \RUSM{} where $(\alpha(\beta),\beta)$ is inapproximable for any $(\alpha(\beta),\beta)$ in \Cref{tab:inapprox-usm}. In particular, $(0.440,1)$ is inapproximable.

\end{restatable}

\begin{table}[h]
    \centering
    \begin{tabular}{|c|c|c|c|c|}
    \hline
    $\beta$ & $\alpha(\beta)$ & $\kappa$ & $\ell_p$ & $\ell_q$ \\
    \hline
    0.1 & 0.0935 & 0.6705 & -0.6095 & -0.2680 \\
    0.2 & 0.1743 & 0.6513 & -0.5322 & -0.2192 \\
    0.3 & 0.2433 & 0.6498 & -0.4705 & -0.1505 \\
    0.4 & 0.3008 & 0.6506 & -0.4207 & -0.0893 \\
    0.5 & 0.3477 & 0.6484 & -0.3800 & -0.0410 \\
    0.6 & 0.3846 & 0.6405 & -0.3400 & 0.0020 \\
    0.7 & 0.4114 & 0.6288 & -0.2900 & 0.0600 \\
    0.8 & 0.4295 & 0.6099 & -0.2400 & 0.1200 \\
    0.9 & 0.4384 & 0.6092 & -0.2100 & 0.1700 \\
    1.0 & 0.4392 & 0.5888 & -0.1800 & 0.2100 \\
    \hline
    \end{tabular}
    \caption{Inapproximability of $(\alpha(\beta),\beta)$-approximations for \texttt{RegularizedUSM} with unconstrained $\ell$ (\Cref{thm:inapprox-usm})}
    \label{tab:inapprox-usm}
\end{table}

\begin{proof} % [Proof of \Cref{thm:inapprox-usm}]

The construction is the same as \Cref{thm:inapprox-usm-non-pos} but we allow both $\ell_p$ and $\ell_q$ to be positive. Therefore, our goal is to compute:
\begin{equation*}
\min_{0\le \kappa\le 1,\ell_p,\ell_q}\brac{\max_{0\le q\le 1, 0\le p}\brac{(1-\kappa)2q(1-q)+\kappa 2(1-q)(1-e^{-p})+2p\ell_p+2q\ell_q}-\beta (\ell_p+\ell_q)}.
\end{equation*}
It turns out that allowing $\ell_q$ to be positive gives better bounds than \Cref{thm:inapprox-usm-non-pos} for $\beta > 0.6$.

\end{proof}

We can do slightly better than \Cref{thm:inapprox-usm} for $\beta$ very close to one with a construction inspired by \cite[Theorem 1.6]{bodek2022maximizing}.

\begin{restatable}[Inapproximability of \RUSM{}, $\beta=1$]{theorem}{inapproxRusmBetter}
\label{thm:0.408-inapprox}

There are instances of \RUSM{} where $(0.408,1)$ is inapproximable.

\end{restatable}

\begin{proof} % [Proof of \Cref{thm:0.408-inapprox}]

As usual, we use \Cref{lemma:symmetry-gap-inapprox-linear}. Define $f$ to be the directed cut function of a generalized hyperedge $(a_1\dots a_k, b_1\dots b_k)$; that is, the generalized hyperedge is said to be cut by $S$ if $S$ contains at least one of the tails of the hyperedge ($a_1 \dots a_k$) but not all of the heads of the hyperedge ($b_1\dots b_k$).
\[\NN\triangleq \{a_1 \dots a_k, b_1\dots b_k\}\]
\[f(S)\triangleq [S\cap \{a_1\dots a_k\} \neq \emptyset]\cdot [\{b_1 \dots b_k\}\not \subset S]\]
\[\ell(a_i)=-0.2037, \ell(b_i)=0.2037\]
Define $p\triangleq \sum_{i=1}^k\bx_{a_i}$ and $q\triangleq k-\sum_{i=1}^k\bx_{b_i}$. Then as $k\to\infty$,
\[F(\overline{\bx})=\paren{1-\paren{1-p/k}^k}\paren{1-\paren{1-q/k}^k}\approx (1-e^{-p})(1-e^{-q}).\]
Now,
\begin{align*}
\max_{\bx}[F(\overline \bx)+L(\overline \bx)]&=\max_{p,q\ge 0}[(1-e^{-p})(1-e^{-q})-0.2037(p+q)+0.2037k]\\
&=0.2037k,
\end{align*}
where the last equality follows since the maximum is attained at $p=q=0$. On the other hand,
\[
\max_S[\alpha f(S)+\ell(S)]\ge (f+\ell)(\{a_1,b_{1\dots k-1}\})=\alpha+0.2037(k-2).
\]
It follows that we have shown $(\alpha,1)$-inapproximability for any $\alpha$ satisfying
\[
0.2037k<\alpha+0.2037(k-2)\implies \alpha > 0.4074. \qedhere
\]

\end{proof}

\section{Open Problems}\label{sec:open}

Most of these problems pertain to \RUSM{} because the gaps between approximability and inapproximability are larger for \RUSM{} compared to \RCSM{}.

\subsection{Approximability}

\paragraph*{\Cref{sec:non-positive}: Non-positive $\ell$.}

\Cref{thm:rusmNonposExtended} attains bounds for \RUSM{} with $\alpha\ge 0.398$. What is the maximum $\alpha$ such that an $(\alpha,\beta)$ approximation exists for some $\beta$? In particular, is $\alpha=0.5$ achievable?

\paragraph*{\Cref{sec:unconstrained}: Unconstrained $\ell$.} 

Is there an algorithm that achieves an $(\eps,1)$ approximation for \RUSM{}? Recall that for \RCSM{} this was achievable when $\ell$ was restricted to be non-positive or non-negative (\Cref{thm:0.385-non-pos,thm:rusmNonnegBetaOne}, respectively), but not in the case where $\ell$ can have arbitrary sign (\Cref{thm:inapprox-csm-beta-1}).

\paragraph*{\Cref{subsec:online-algos}: Online \RUSM{}.}

Can the $\alpha$ in \Cref{thm:dicut-oblivious} be improved? Is there an online algorithm that works for general non-monotone $f$ and $\ell$ non-positive? We note that the \textit{semi-streaming} algorithms studied by Kazemi et al. \cite{kazemi2021regularized} and Nikolakaki et al. \cite{nikolakaki2021efficient} provide a $(0.5,1)$-approximation algorithm for \RUSM{} when $f$ is monotone. For non-monotone \USM{}, simply selecting each element of $f$ with probability $0.5$ achieves a $0.25$-approximation \cite{feige2011maximizing}. For non-monotone \CSM{} where the constraint is a cardinality constraint, Buchbinder et al. \cite{buchbinder2014online} provide an online algorithm achieving a competitive ratio of $\frac{56}{627}>0.0893$ when \textit{preemption} is allowed.

\subsection{Inapproximability}

\Cref{tab:best-inapprox} summarizes some of the best known inapproximability results and their corresponding approximation guarantees. The gaps between approximability and inapproximability are particularly large in the second and fourth rows, corresponding to \RUSM{} for $\ell\le 0$ and unconstrained $\ell$, respectively. All inapproximability results use the symmetry gap technique; are there any other inapproximability techniques potentially worth considering?

\begin{table}[h]
    \centering
    \begin{tabular}{|c|c|c|c|c|}
    \hline 
    Section & Problem & Inapproximability & Approximability & Source of Inapproximability \\
    \hline 
    \Cref{sec:cardinality-inapprox} & \CSM{} & 0.478 & 0.385 & \cite[Theorem E.2]{gharan2011submodular} \\
    \Cref{sec:non-positive} & \RUSM{}, $\ell\le 0$ & $(0.5, 2\ln 2-\eps)$ & None for $\alpha=0.5$ & \Cref{thm:inapprox-two-ln-two} \\
    \Cref{sec:non-neg-usm} & \RUSM{}, $\ell\ge 0$ & $\paren{0.5,2\sqrt 2/3+\eps}$ & $(0.5,0.75)$ & \Cref{thm:0.5-0.943-inapprox-nonneg} \\
    \Cref{sec:unconstrained} & \RUSM{} & $\paren{0.408,1}$ & $(0,1)$ & \Cref{thm:0.408-inapprox} \\
    \hline 
    \end{tabular}
    \caption{Gaps Between Current Approximability and Inapproximability}
    \label{tab:best-inapprox}
\end{table}

\section*{Acknowledgements}

This research was conducted as part of MIT's Undergraduate Research Opportunities Program. I thank my supervisor Tasuku Soma for many helpful discussions, as well as the authors of \cite{bodek2022maximizing}, which this work was heavily influenced by.

\printbibliography

\appendix

\section{Appendix}

\subsection{Prior Work}\label{subsec:prior-work}

We outline the general idea for all \textit{continuous greedy} algorithms because our results build on them.

\subsubsection{Submodular Maximization}

For surveys of submodular maximization results, see Krause and Golovin \cite{krause2014submodular} or Buchbinder and Feldman \cite{buchbinder2018submodular}. 

\paragraph*{$f$ Monotone (Constrained):}

It is well-known that a simple \textit{greedy} algorithm achieves a $\paren{1-\frac{1}{e}}$-approximation for maximizing monotone submodular functions subject to a cardinality constraint \cite{nemhauser1978analysis}, and that this approximation factor is optimal \cite{nemhauser1978best}.

Calinescu et al. \cite{calinescu2011maximizing} introduced the \textit{continuous greedy} algorithm, which achieves a $\paren{1-\frac{1}{e}}$-approximation for maximizing the multilinear extension of a monotone submodular function over a solvable down-closed polytope $\PP$. The idea is to continuously evolve a fractional solution $\by(t)$ from ``time'' $t=0$ to $t=1$ such that 
\[\by(t)\in t\cdot \PP \qquad \text{and}\qquad F(\by(t))\ge (1-e^{-t})f(OPT).\]
This continuous process can be discretized into a polynomial number of steps at the cost of a negligible loss in the approximation factor. If $\PP$ is the \textit{matroid polytope} corresponding to a matroid $\MM=(\NN,\II)$, then \textit{pipage rounding} may be used to round the fractional solution $\by(1)$ to an independent set $S\in \II$ such that $\EE[f(S)]\ge F(\by(1))$ \cite{vondrak2011symmetry}.

\paragraph*{$f$ Non-monotone (Unconstrained):} Feige et al. \cite{feige2011maximizing} showed that no polynomial-time algorithm may provide a $(0.5+\eps)$-approximation for maximizing a non-monotone submodular function. Buchbinder et al. \cite{buchbinder2012double} later discovered a \textit{randomized double greedy} algorithm that achieves a 0.5-approximation in expectation. The idea is to iterate through the elements of the ground set $\NN$ in arbitrary order, and for each one choose whether or not to include it in the returned set with some probability.

\paragraph*{$f$ Non-monotone (Constrained):}

Feldman et al. \cite{feldman2011unified} showed a $1/e>0.367$-approximation for maximizing the \textit{multilinear extension} of a non-monotone submodular function over a solvable down-closed polytope $\PP$ using a \textit{measured continuous greedy}. The idea is to continuously evolve a fractional solution $\by(t)$ from $t=0$ to $t=1$ such that 
$$\by(t)\in (t\cdot \PP)\cap ((1-e^{-t})\cdot [0,1]^{\NN}) \qquad \text{and}\qquad F(\by(t))\ge te^{-t}f(OPT).$$
As with the original continuous greedy, the fractional solution $\by(1)$ can be rounded to an integer solution when $\PP$ is a matroid polytope. Additionally, when $f$ is monotone, measured continuous greedy provides the same guarantee as \cite{calinescu2011maximizing}.

The approximation factor was later improved by Buchbinder and Feldman \cite{buchbinder2016nonsymmetric} to $0.385$. The idea is to first run \textit{local search} on the multilinear extension $F$ to find a ``locally optimal'' fractional solution $\bz\in \PP$, round $\bz$ to a set $Z$, and then run a measured continuous greedy ``aided'' by $Z$. Either $Z$ will be a 0.385-approximation in expectation, or the set returned by \textit{aided measured continuous greedy} will be. The aided measured continuous greedy consists of running measured continuous greedy from $t=0$ to $t=t_s$ on $\NN\backslash Z$, followed by running measured continuous greedy from $t=t_s$ to $t=1$ on the entire ground set $\NN$, where $t_s=0.372$. The optimal value of $t_s$ was determined by solving a non-convex optimization problem.

On the inapproximability side, Gharan and Vondrak \cite{gharan2011submodular} showed that no polynomial-time algorithm may achieve a 0.478-approximation for maximizing a non-negative submodular function subject to a matroid independence constraint or a 0.491-approximation for maximizing a non-negative submodular function subject to a cardinality constraint using the \textit{symmetry gap} framework of Vondrak \cite{vondrak2011symmetry}. The symmetry gap framework may also be used to succinctly reprove the optimality of the $1-\frac{1}{e}$ and $\frac{1}{2}$ approximation factors for monotone and nonmonotone maximization, respectively, which were previously proved by ad hoc methods. The idea is that given a maximization problem with a symmetry gap of $\gamma\in (0,1)$, we can construct a family of pairs of functions that require exponentially many value oracle queries to distinguish but whose maxima differ by a factor of $\gamma$. This in turn shows the inapproximability of a $(\gamma+\eps)$-approximation.

\subsubsection{Regularized Submodular Maximization}\label{subsubsec:regularized-submod}

\paragraph*{Monotone $f$:} Sviridenko et al. \cite{sviridenko2017optimal} first presented an $(1-1/e-\eps,1-\eps)$-approximation algorithm for \texttt{RegularizedCSM} involving a step where the value of $\ell(OPT)$ needs to be ``guessed'' to within a factor of $1+\eps$, followed by continuous greedy on $\PP\cap \{\bx: L(\bx)\ge \ell(OPT)\}$. Afterward, if $\PP$ is a matroid independence polytope, $\bx$ can be rounded to a set $S$ such that $\bone_S\in \PP$ using pipage rounding such that $\EE[f(S)+\ell(S)]\ge F(\bx)+L(\bx)$.

Feldman \cite{feldman2018guess} later combined continuous greedy with the notion of a \textit{distorted objective} that initially places higher weight on the linear term and increases the weight on the submodular term over time. This \textit{distorted continuous greedy} achieves the same approximation factor as \cite{sviridenko2017optimal} without the need for the guessing step. The idea is to continuously evolve a fractional solution $\by(t)$ from $t=0$ to $t=t_f$ such that
\[\by(t)\in t\cdot \PP \qquad \text{and}\qquad G_t(\by(t))\ge (e^{t-t_f}-e^{-t_f})f(OPT)+t\ell(OPT),\]
where $G_t(\by)\triangleq e^{t-t_f}F(\by)+L(\by)$ is the distorted objective at time $t$.\footnote{Actually, the original paper shows this only for $t_f=1$, but as noted by \cite{bodek2022maximizing} this can easily be generalized.} For $t_f=1$, this gives a $(1-1/e-\eps,1)$-approximation, eliminating the $\eps$ in the linear term that appears in the bound of \cite{sviridenko2017optimal} due to the guessing step.

Using the symmetry gap technique \cite{vondrak2011symmetry}, Bodek and Feldman \cite[Theorem 1.1]{bodek2022maximizing} proved that no $(1-e^{-\beta}+\eps,\beta)$-approximation algorithm for \texttt{RegularizedUSM} exists for any $\beta\ge 0$, even when $\ell$ is constrained to be non-positive (see \Cref{fig:ell-non-pos} for an illustration). This matches the guarantee of \textit{distorted continuous greedy}, which achieves a $(1-e^{-\beta}-\eps,\beta)$-approximation for \texttt{RegularizedCSM} whenever $\beta\in [0,1]$. When $\ell$ is constrained to be non-positive, Lu et al. \cite{lu2021regularized} achieve a $(1-e^{-\beta}-\eps,\beta)$-approximation for \texttt{RegularizedCSM} for any $\beta\ge 0$ using \textit{distorted measured continuous greedy} (described below). For the remainder of this section, $f$ is not necessarily monotone.

\paragraph*{Non-positive $\ell$:} Lu et al. \cite{lu2021regularized} presented a $(\beta e^{-\beta}-\eps,\beta)$-approximation algorithm for \texttt{RegularizedCSM} combining the measured and distorted continuous greedies mentioned above due to Feldman et al. \cite{feldman2011unified,feldman2018guess}. The idea is to continuously evolve a solution $\by(t)$ from $t=0$ to $t=t_f$ such that
\[\by(t)\in (t\cdot \PP)\cap ((1-e^{-t})\cdot [0,1]^{\NN}) \qquad \text{and}\qquad G_t(\by(t))\ge te^{-t_f}f(OPT)+t\ell(OPT),\]
where $G_t(\by)=e^{t-t_f}F(\by)+L(\by)$ as in distorted continuous greedy above.\footnote{Actually, the original paper shows this only for $t_f=1$, but as noted by \cite{bodek2022maximizing} this can easily be generalized.} Setting $t_f=\beta$ gives the desired approximation factor. Note that when $\ell=0$, the guarantee of \textit{distorted measured continuous greedy} becomes the same as \textit{measured continuous greedy}. As noted in the previous paragraph, the approximation guarantee of this algorithm becomes the same as Feldman's distorted continuous greedy when $f$ is monotone.

Bodek and Feldman \cite[Theorem 1.3]{bodek2022maximizing} proved $(\alpha(\beta),\beta)$-inapproximability for \texttt{RegularizedUSM} for all $\beta\ge 0$, where $\alpha(\beta)$ is an increasing function satisfying $\alpha(1)\approx 0.478$, matching the best known bound for maximizing a submodular function subject to a matroid constraint \cite{gharan2011submodular} (see \Cref{fig:ell-non-pos} for an illustration). 

\paragraph*{Non-negative $\ell$:} Bodek and Feldman \cite[Theorem 1.5]{bodek2022maximizing} showed that Buchbinder et al.'s \textit{double greedy} \cite{buchbinder2012double} is simultaneously a $(\alpha,1-\alpha/2)$-approximation for \texttt{RegularizedUSM} for any $\alpha\in [0,0.5]$, and that a $(0.4998+\eps,1)$-approximation for \texttt{RegularizedUSM} is impossible \cite[Theorem 1.6]{bodek2022maximizing}.

\paragraph*{Unconstrained $\ell$:} Bodek and Feldman \cite[Theorem 1.2]{bodek2022maximizing} presented a $\paren{\frac{\beta(1-\beta)}{1+\beta}-\eps,\beta}$-approximation for \texttt{RegularizedUSM} using a \textit{local search} technique.

Sun et al. \cite{sun2022maximizing} presented an algorithm for \texttt{RegularizedCSM} where the sign of $\ell$ is unconstrained which turns out to be identical to that of Lu et al. \cite{lu2021regularized}. They showed that their algorithm outputs $\bx\in \PP$ such that $F(\bx)+L(\bx)\ge \max_{S\in \PP}\brac{\paren{\frac{1}{e}-\eps}\cdot f(S)+\paren{\frac{\beta(S)-e}{e(\beta(S)-1)}}\ell(S)}$, where $\beta(S)\triangleq \frac{\sum_{u\in S\cap \NN^+}\ell(u)}{-\sum_{u\in S\cap \NN^-}\ell(u)}\ge 0$. Note that when $\ell\le 0$, $\beta(S)=0$ and the coefficient of $\ell(S)$ is $1$, recovering the approximation guarantee of Lu et al. \cite{lu2021regularized}. However, this is not quite an $(\alpha,\beta)$-approximation algorithm when $\ell$ is allowed to have arbitrary sign since $\beta(S)$ is not constant. Furthermore, the expression $\frac{\beta(S)-e}{e(\beta(S)-1)}$ could potentially be negative, which is problematic.

\subsection{Online Algorithms for \RUSM{}}\label{subsec:online-algos}

Here, we discuss whether the general class of \textit{online} algorithms can achieve approximation factors for \RUSM{} when $\ell$ is not necessarily non-negative. First, we formally define the notion of online algorithms in the context of $f+\ell$ sums with $f$ a directed cut function.

\begin{definition}[Online Algorithms for Directed Cuts]

When $f$ is a directed cut function, we say that an algorithm is \textit{online} in the sense of Bar-Noy and Lampis \cite{bar2012online} if it works in the following setting:
\begin{enumerate}
    \item The vertices of the ground set are revealed in the order $u_1,u_2,\dots,u_n$.
    \item The algorithm is provided with $\ell_u$, the total in-degree of $u$ ($in(u)$), the total out-degree of $u$ ($out(u)$), as well as the edges between $u$ and all previously revealed vertices, only after $u$ is revealed.
    \item The algorithm makes an irreversible decision about whether to include $u$ in the returned set before any vertices after $u$ are revealed.
\end{enumerate}

\end{definition}

We note that both \DetDG{} and \RanDG{} are examples of online algorithms. Huang and Borodin \cite{huang2014bounds} extended the notion of online algorithms to general non-monotone submodular $f$, though we do not consider their extension here. We first show that deterministic online algorithms cannot achieve any $(\alpha,\beta)$-approximation.

\begin{theorem}\label{thm:dicut-deterministic-bad}

There are instances of \RUSM{} with $f$ a directed cut function and $\ell$ non-positive such that no deterministic online algorithm can provide a $(\alpha, \beta)$-approximation for any $\alpha>0$.

\end{theorem}

\begin{proof} % [Proof of \Cref{thm:dicut-deterministic-bad}]

Suppose that after $u_1$ is revealed, the algorithm is provided with $in(u_1)=1$, $out(u_1)=\alpha/2$, and $\ell(u_1)=0$. 
\begin{enumerate}
    \item If the algorithm includes $u_1$ in the returned set, then this algorithm fails to provide the desired approximation factor on the following instance:
    \[f(S)=\alpha/2\cdot [u_1\in S\text{ and }u_2\not\in S]+[u_2\in S\text{ and }u_1\not \in S], \ell(u_2)=0,\] 
    since it outputs a set with value at most $\alpha/2$, whereas if we let $OPT\triangleq \{u_2\}$ then $f(OPT)=1$ and $\ell(OPT)=0$, implying that $\alpha f(OPT)+\beta \ell(OPT)=\alpha > \alpha/2$.
    
    \item On the other hand, if the algorithm does not include $u_1$ in the returned set, then this algorithm fails to provide the desired approximation factor on the following instance:
    \[f(S)=\alpha/2\cdot [u_1\in S\text{ and }u_2\not\in S]+[u_2\in S\text{ and }u_1\not \in S], \ell(u_2)=-1,\] 
    since it outputs a set with value $0$ whereas if we let $OPT\triangleq \{u_1\}$ then $f(OPT)=\alpha/2$ and $\ell(OPT)=0$, implying that $\alpha f(OPT)+\beta \ell(OPT)=\alpha/2>0$.
\end{enumerate}

\end{proof}

We next show that \RanDG{} does not achieve any $(\alpha,\beta)$-approximation by adapting the proof of \Cref{thm:randomized-dg-tight}.

\begin{corollary}\label{corollary:randomized-dg-bad}

There are instances of \RUSM{} with $f$ a directed cut function and $\ell$ non-positive such that \RanDG{} does not provide any $(\alpha,\beta)$-approximation for any $\alpha>0$.

\end{corollary}

\begin{proof} % [Proof of \Cref{corollary:randomized-dg-bad}]

Define $f$ to be the same as in \Cref{thm:randomized-dg-tight}, $\ell(u_n)=0$, and $\ell(u_i)=\frac{1-r}{n-1}$ for all $i\in [1,n-1]$. Then
\begin{equation*}
\max_S[\alpha f(S)+\beta \ell(S)]\ge (\alpha f+\beta \ell)(\{u_n\})=\alpha\cdot 1+\beta\cdot 0=\alpha.
\end{equation*}
For each $i\in [1,n-1]$, \[a_i=\ell(u_i)+f(u_i|\emptyset)=\frac{1-r}{n-1}+\frac{r}{n-1}=\frac{1}{n-1}\]
and 
\[b_i=-\ell(u_i)-f(u_i|\NN\backslash \{u_i\})=\frac{r-1}{n-1}+\frac{1}{n-1}=\frac{r}{n-1}.
\]
So by similar reasoning as the proof of \Cref{thm:randomized-dg-tight}, the fraction $f$ of $\{u_1,\dots,u_{n-1}\}$ selected by double greedy will be close to $\frac{1}{r+1}$ with high probability. If $u_n$ is included in the returned set, then the value of the set will be $f(1-r)+(1-f)\approx \frac{1}{r+1}$, whereas if $u_n$ is not, then the value of the set will be
$f\approx \frac{1}{r+1}$. So regardless of whether double greedy chooses to include $u_n$ in the returned set or not, the returned set will have expected value at most $\frac{1}{r+1}+\eps<\alpha$ when both $r$ and $n$ are sufficiently large.

\end{proof}

On the other hand, there are randomized algorithms that achieve $(\alpha,\beta)$-approximations. In fact, the algorithm we provide next is \textit{oblivious} in the sense of Feige and Shlomo \cite{feige2015oblivious}; that is, it uses only information local to each vertex.

\begin{theorem}\label{thm:dicut-oblivious}

For any $\beta\in [0,1]$, there is an oblivious $(\beta(1-\beta),\beta)$-approximation algorithm for \RUSM{} with $f$ a directed cut function and $\ell$ having arbitrary sign. 

\end{theorem}

\begin{proof} % [Proof of \Cref{thm:dicut-oblivious}]

For each vertex $v$, select it with probability $\beta$ if $(1-\beta)\cdot \text{out}(v)+\ell(v)\ge 0$, and 0 otherwise. Then
\begin{align*}
\beta(1-\beta) f(OPT)+ \beta\ell(OPT)&\le \sum_{v\in OPT}\paren{\beta(1-\beta)\cdot \text{out}(v)+\beta \ell(v)}\\
&\le \sum_{v\in V}\beta\max((1-\beta)\cdot \text{out}(v)+\ell(v),0),    
\end{align*}
and the last expression lower bounds the expected value of the solution returned by the randomized algorithm since every vertex is \textit{not} selected with probability at least $1-\beta$.

\end{proof}

\subsection{Omitted Proofs}\label{subsec:omitted-proofs}

\begin{proof}[Proof of \Cref{lemma:aided-ext} (Remainder)]
We modify the non-formal proof of \cite{buchbinder2016nonsymmetric}. This non-formal proof uses some simplifications such as allowing a direct oracle access to the multilinear extension $F$ and giving the algorithm in the form of a continuous time algorithm, but these simplifications may be removed using known techniques at the cost of introducing the $o(1)$ into the guarantee \cite{buchbinder2016nonsymmetric}. 

By \cite[Lemma 4.3]{buchbinder2016nonsymmetric},
\begin{equation*}
\frac{dF(\by(t))}{dt}\ge \begin{cases}
F(\by(t)\vee \bone_{OPT\backslash Z})-F(\by(t)) & t\in [0,t_s) \\
F(\by(t)\vee \bone_{OPT})-F(\by(t)) & t\in [t_s,t_f)
\end{cases}.
\end{equation*}
Also, by \cite[Lemma 4.4]{buchbinder2016nonsymmetric}, for every time $t\in [0,t_f)$ and set $A\subseteq \NN$ it holds that:
\begin{equation*}
F(\by(t)\vee \bone_A)\ge \paren{e^{-\max\{0,t-t_s\}}-e^{-t}}[f(A)-f(A\cup Z)]+e^{-t}f(A).
\end{equation*}
So then by \cite[Corollary 4.5]{buchbinder2016nonsymmetric}, plugging in $A=OPT\backslash Z$ and $A=OPT$ for $t\in [0,t_s)$ and $t\in[t_s,t_f)$, respectively, gives us
\begin{align*}
\frac{dF(\by(t))}{dt}&\ge \begin{cases}
f(OPT\backslash Z)-(1-e^{-t})f(OPT\cup Z) & t\in [0,t_s) \\
e^{t_s-t}f(OPT)-(e^{t_s-t}-e^{-t})f(OPT\cup Z) & t\in [t_s,t_f)
\end{cases}-F(\by(t))\\
&\triangleq G(t)-F(\by(t)).
\end{align*}
By submodularity of $f$, we may replace $f(OPT\backslash Z)$ with $f(OPT) - f(OPT\cap Z)$ in $G(t)$. Then
\begin{align*}
F(\by(t_f))&\ge e^{-t_f}\brac{\int_0^{t_f}e^tG(t)\,dt+F(\by(0))}\\
&=e^{-t_f}\left(\int_0^{t_s}e^t\brac{f(OPT)-f(OPT\cap Z)-(1-e^{-t})f(OPT\cup Z)}dt\right.\\
&\left.+\int_{t_s}^{t_f}e^t\brac{e^{t_s-t}f(OPT)-(e^{t_s-t}-e^{-t})f(OPT\cup Z)}dt\right).
\end{align*}
After evaluating and rearranging this final expression, we can see that this matches \Cref{align:aided-ext-integral}.

\end{proof}

\begin{proof}[Proof of \Cref{thm:max-cut-csm}]
Consider the following linear program:
\begin{maxi}
{\bx\in \PP}{\frac{1}{2}\hat{f}(\bx)+L(\bx)\triangleq  \frac{1}{2}\max_{\bf c}\paren{\sum_{ab\in E}w_{ab}{\bf c}_{ab}}+L({\bf x})}{}{\label{maxi:undirected}}
\addConstraint{{\bf c}_{ab}}{\ge 0}
\addConstraint{{\bf c}_{ab}}{\le \bx_a+\bx_b}
\addConstraint{{\bf c}_{ab}}{\le 2-\bx_a-\bx_b}
\end{maxi}
Here, ${\bf c}_{ab}$ corresponds to whether the edge $(a,b)$ was cut. Note that $\hat{f}(\bone_S)=f(S)$ for all $S\subseteq \NN$, meaning that $\hat{f}$ is an extension of $f$ (though not multilinear).  Furthermore, since $f$ is an \textit{undirected} cut function,
\begin{equation}
\forall (\bx_a, \bx_b)\in [0,1]^2, \bx_a(1-\bx_b)+(1-\bx_a)\bx_b\ge \frac{1}{2}\min(\bx_a+\bx_b,2-\bx_a-\bx_b),\label{eq:undirected-relax}
\end{equation}
implying that $F(\bx)\ge \frac{1}{2}\hat{f}(\bx)$ for all $\bx\in [0,1]^{\NN}$. \Cref{eq:undirected-relax} can be verified by first replacing $(\bx_a,\bx_b)$ with $(1-\bx_a,1-\bx_b)$ if $\bx_a+\bx_b>1$, and then performing the following sequence of computations:
\begin{align*}
\bx_a(1-\bx_b)+(1-\bx_a)\bx_b &= \bx_a+\bx_b-2\bx_a\bx_b \\
&\ge \bx_a+\bx_b-\frac{(\bx_a+\bx_b)^2}{2}\\
&=   \bx_a+\bx_b\paren{1-\frac{\bx_a+\bx_b}{2}}\\
&\ge \frac{\bx_a+\bx_b}{2}.
\end{align*}
Let $\bx^*$ be a solution attaining the optimal value for \Cref{maxi:undirected}, which can be found using any LP solver (e.g. using the ellipsoid method). Then
\begin{align*}
F(\bx^*)+L(\bx^*)&\ge \frac{1}{2}\hat{f}(\bx^*)+L(\bx^*)\\
&=\max_{\bx\in \PP}\brac{\frac{1}{2}\hat{f}(\bx^*)+L(\bx^*)}\\
&=\max_{S\in \II}\brac{\frac{1}{2}f(S)+\ell(S)}\\
&\ge \frac{1}{2}f(OPT)+\ell(OPT).
\end{align*}
Thus, $\bx^*$ achieves the desired approximation factor. We can finish by using pipage rounding to round $\bx^*$ to an integral solution within $\II$ that preserves the value of $f+\ell$ in expectation.
\end{proof}

\begin{proof}[Proof of \Cref{thm:max-dicut-usm}]
Consider a linear program similar to the one in the proof of \Cref{thm:max-cut-csm}.
\begin{maxi}
{\bx\in [0,1]^{\NN}}{\frac{1}{2}\hat{f}(\bx)+L(\bx)\triangleq  \frac{1}{2}\max_{\bf c}\paren{\sum_{ab\in E}w_{ab}{\bf c}_{ab}}+L({\bf x})}{}{\label{maxi:dicut}}
\addConstraint{{\bf c}_{ab}}{\ge 0}
\addConstraint{{\bf c}_{ab}}{\le {\bx}_a}
\addConstraint{{\bf c}_{ab}}{\le 1-{\bx}_b}
\end{maxi}
Unfortunately, it does \textit{not} suffice to just find any $\bx^*$ that attains the optimum value and apply pipage rounding. The reason for this is that $F(\bx)\not \ge \frac{1}{2}\hat{f}(\bx)$ in general. However, it can be verified that 
\begin{enumerate}
    \item $F(\bx)\ge \frac{1}{2}\hat{f}(\bx)$ when $\bx$ is \textit{half-integral}; that is, $\bx_u\in \{0,0.5,1\}$ for all $u\in \NN$. This inequality can easily be verified for the cut function of a single directed edge, and thus extends to sums of cut functions.
    
    \item The $(\bx,\bf c)$ polytope defined by the constraints in \Cref{maxi:dicut} is bounded, and all its vertices are half-integral. Note that this property would no longer hold if the constraint $\bx\in \PP$ was included, which is why \Cref{thm:max-dicut-usm} does not apply to \RCSM{}.
\end{enumerate}
Both of these properties were previously used by Halperin and Zwick's combinatorial 0.5-approximation to \texttt{MAX-DICUT} \cite{halperin2001combinatorial}. Thus, the remainder of our algorithm is identical to that in the proof of \Cref{thm:max-cut-csm}, except we additionally require that the point $\bx^*$ returned by the LP solver is a vertex of the polytope to guarantee that $F(\bx)\ge \frac{1}{2}\hat{f}(\bx)$.

\end{proof}

\begin{proof}[Proof of \Cref{thm:dicut-inapprox}]
Our goal is to show that when $f$ is a directed cut function,
\[\max_{\bx \in [0,1]^{\NN}}\brac{F(\overline \bx)+L(\overline \bx)}\ge \max_{S\subseteq \NN}[0.5f(S)+\ell(S)].\]
The idea is to first construct an auxiliary function $\hat{f}(\bx)$ satisfying the following properties (note that we do not capitalize $\hat{f}$ since it is not a multilinear extension):
\begin{enumerate}
    \item The function $\hat{f}$ is symmetric; that is, $\hat{f}(\bx)=\hat{f}(\overline{\bx})$ for all $\bx \in [0,1]^{\NN}$.
    \item The function $\hat{f}$ upper bounds $f$; that is, $\hat{f}(\bone_S)\ge f(S)$ for all $S\subseteq \NN$.
    \item There exists $\bx^*\in \text{argmax}_{\bx\in [0,1]^{\NN}}\brac{0.5 \hat{f}(\bx)+L(\bx)}$ such that $\bx^*=\overline{\bx^*}$ and $F(\bx^*)\ge 0.5\hat{f}(\bx^*)$.
\end{enumerate}
Assuming that all these properties hold, we find:
\begin{align*}
\max_{\bx\in [0,1]^{\NN}} \brac{F(\overline{\bx})+L(\overline{\bx})}
& \ge \max_{\bx\in [0,1]^{\NN}} \brac{0.5 \hat{f}(\bx)+L(\bx)} && \text{(by properties 1 and 3)} \\
& \ge \max_{S\subseteq \NN}[0.5 \hat{f}(\bone_S)+\ell(S)] && \\
& \ge \max_{S\subseteq \NN}\brac{0.5 f(S)+\ell(S)} && \text{(by property 2)}.
\end{align*}

Before defining $\hat{f}$, we examine the symmetrization operator $\overline{\bx}$. Recall that symmetrization is defined with respect to a permutation group $\GG$. Partition the ground set into $K$ subsets $\NN=\NN_1\cupdot \NN_1 \dots \cupdot \NN_K$, where $\cupdot$ denotes the disjoint union of two sets, and define the function $g: \NN \to \{1,2,\dots,K\}$ to be the mapping from every element of the ground set to the subset that contains it. This mapping satisfies the property that $g(u_i)=g(u_j)$ if and only if there exists a permutation $\sigma \in \GG$ such that $\sigma(u_i)=u_j$. Observe that \[\overline{\bx}_{u_i}=\text{avg}_{g(u_i)}(\bx)\triangleq \frac{\sum_{u\in \NN_{g(u_i)}}\bx_u}{|\NN_{g(u_i)}|};\] 
that is, the value at $u_i$ in $\overline{\bx}$ is just the average of the values in $\bx$ of all $u$ in the same subset as $u_i$.

Next, we define $\hat{f}$ in terms of $\text{avg}_1(\bx),\text{avg}_2(\bx),\dots,\text{avg}_K(\bx)$, which guarantees that property 1 is satisfied. For all $1\le i,j\le K$, define $w_{ij}\ge 0$ as the sum of the weights of the edges directed from $\NN_i$ to $\NN_j$ (where $i$ can equal $j$). Then
\[\hat{f}(\bx)\triangleq \sum_{i=1}^K\sum_{j=1}^Kw_{ij}\min(\text{avg}_i(\bx),1-\text{avg}_j(\bx)).\]
It remains to show that properties 2 and 3 are satisfied.

\paragraph*{Property 2:} Since every subset $\NN_i$ is symmetric, the proportion of edges from $\NN_i$ to $\NN_j$ that are cut by $S$ is bounded above by the proportion of elements in $\NN_i$ contained within $S$ (which is precisely $\text{avg}_i(\bx)$) as well as the proportion of elements in $\NN_j$ not contained within $S$ (which is precisely $1-\text{avg}_j(\bx)$). This implies that $\hat{f}(\bone_S)\ge f(S)$ for all $S$.

\paragraph*{Property 3:} Similar to the proof of \Cref{thm:max-dicut-usm}, it can be shown that there exists a half-integral tuple $(\text{avg}_1(\bx),\allowbreak\text{avg}_2(\bx),\allowbreak\dots,\text{avg}_K(\bx))$ at which $0.5\hat{f}(\bx)+L(\bx)$ is maximized. This tuple corresponds to a half-integral and symmetric $\bx^*$. As in \Cref{thm:max-dicut-usm}, it is easy to check that $F(\bx^*)\ge 0.5\hat{f}(\bx^*)$, completing the proof.
\end{proof}
\end{document}